\documentclass{article}
\usepackage{breakurl}             % Not needed if you use pdflatex only.

\usepackage{color}
\usepackage{pstricks,pst-node,pst-tree,pstricks-add}
\usepackage{multido}
\usepackage{enumerate}
\usepackage[english]{babel}
\usepackage{amsthm}
\usepackage{amssymb}
\usepackage{amsmath}
\usepackage{amsfonts}
\usepackage[footnote,marginclue]{fixme}

\theoremstyle{plain}
\newtheorem{theorem}{Theorem}
\newtheorem{lemma}[theorem]{Lemma}
\newtheorem{corollary}[theorem]{Corollary}
\newtheorem{proposition}[theorem]{Proposition}

\theoremstyle{definition}

\newtheorem{definition}[theorem]{Definition}

\newtheorem{remark}{Remark}

\newcommand{\sg}{\sigma}

\newcommand{\mA}{{\mathfrak A}}
\newcommand{\ma}{{\mathfrak A}}
\newcommand{\Dom}{{\rm Dom}}
\newcommand{\dom}{{\rm dom}}

\newcommand{\restrict}[2]{{#1}({#2})}

\newcommand{\calC}{{\mathcal C}}
\newcommand{\calS}{{\mathcal S}}
\newcommand{\mX}{{\mathfrak X}}

\newcommand{\Fr}{{\rm Fr}}
\newcommand{\Var}{{\rm Var}}

\newcommand{\ESO}{{\rm ESO}}

\newcommand*\dep{{=\mkern-1.2mu}} 

\newcommand{\df}{\rm{D}}
\newcommand{\FOinclusion}{{\rm FO(\subseteq)}}

\newcommand{\ESOfarity}[1]{{\rm ESO}_f({#1}\mbox{-ary})}

\newcommand{\fo}{{\rm FO}}
\newcommand{\np}{{\rm NP}}
\newcommand{\Ptime}{{\rm PTIME}}

\newcommand{\nl}{{\rm NL}}
\newcommand{\logspace}{{\rm L}}

\newcommand{\PGFP}{{\rm GFP}^+}

\newcommand{\tu}[1]{\overline{#1}}

\newcommand{\indep}[3]{{#1}\ \bot_{#2}\ {#3}}

\newcommand{\col}[1]{\mathbf{#1}}

 \newcommand{\pb}[1]{\textsc{#1}}

\newcommand{\Disjunctiondepth}[1]{\textsf{d}_{\vee}(#1)}

\newcommand{\indlogic}{\rm{FO} (\bot_{\rm c})}
\def\one{w_1}
\def\two{w_2}
\def\three{w_3}

\begin{document}

\title{Tractability Frontier of Data Complexity in Team Semantics}

\author{Arnaud Durand\\
IMJ-PRG, CNRS UMR 7586, Universit\'e de Paris, France

\and

Juha Kontinen\\
Department of Mathematics and Statistics,\\
 University of Helsinki, Finland

\and Nicolas de Rugy-Altherre\\
Universit\'e de Lorraine, France

\and
Jouko V\"{a}\"{a}n\"{a}nen\\
Department of Mathematics and Statistics,\\
 University of Helsinki, Finland\\
and
Institute for Logic, Language and Computation,\\
 University of Amsterdam, The Netherlands}
\date{}

\maketitle

%\newpage

%\tableofcontents

\begin{abstract}We study the data complexity of model-checking for logics with team semantics. We focus on dependence, inclusion, and independence logic formulas under both strict and lax team semantics. Our results delineate a clear tractability/intractability frontiers in data complexity of both quantifier-free and quantified formulas for each of the logics. 
For inclusion logic under the lax semantics, we reduce the model-checking problem to the satisfiability problem of so-called \emph{dual-Horn} Boolean formulas. Via this reduction, we give an alternative proof for the known result that the data complexity of inclusion logic is in  PTIME.
\end{abstract}

\section{Introduction}

In this article we study the data complexity of model-checking of dependence, independence, and inclusion logic formulas. Independence  and inclusion logic  \cite{gradel10,galliani12} are variants of  dependence logic \cite{vaananen07}  that  extends first-order logic by dependence atoms of the form 
%\begin{equation}\label{}
$\dep(x_1,\ldots,x_n)$
%\end{equation} 
expressing that the value of $x_n$ is functionally determined by the values of the variables $x_1,\ldots, x_{n-1}$. 
%The semantics of dependence logic is defined using sets of assignments rather than a single assignment as in first-order logic.  
In independence and inclusion logic dependence atoms are replaced by independence and inclusion atoms 
$\tu{y}\ \bot_{\tu{x}} \ \tu{z}$ and  $\tu{x}\subseteq \tu{y},$
respectively. The meaning of the independence atom is that, with respect to any fixed value of  $ \tu x$, the variables $\tu y$  are independent of the variables $\tu z$, whereas the  inclusion atom expresses that all the values of $\tu{x}$ appear also as values for $\tu{y}$.

Team semantics is a  framework for formalizing and studying  various notions of dependence and independence  pervasive in many areas of science.  Team semantics differs from Tarski's semantics by interpreting formulas using sets of assignments instead of single assignments as in  first-order logic.  
 Reflecting this, dependence logic  has higher expressive power than classical logics used for these purposes previously.  Dependence, inclusion, and independence  atoms are intimately connected to the  corresponding functional, inclusion, and multivalued  dependencies studied in  database theory, see, e.g., \cite{DBLP:conf/foiks/HannulaK14}.  Interestingly, independence atoms correspond naturally to a qualitative analogue of the notion of conditional independence in statistics  \cite{DBLP:journals/networks/GeigerVP90}. On the other hand,  team semantics can be naturally generalized to a probabilistic variant in which probabilistic independence can be taken  as an atomic formula (see \cite{HannulaHKKV19,HKMV18} for further details). %Furthermore, a variant of dependence logic is in the heart of  {\em Inquisitive Semantics} which is a novel approach in linguistics that analyzes information exchange through communication, see \cite{DBLP:journals/jphil/CiardelliR11}.
 
 Dependence logic and its variants can be  used to %logically 
 formalize and study dependence and independence notions in various areas. For example, in the foundations of quantum mechanics, there  are a range of notions of independence playing a central role in celebrated No-Go results such as Bell's theorem. 
 %Abramsky and V\"a\"an\"anen have recently showed that, under a  relational view on these results,  some of these No-Go results can be logically formalized and  syntactically derived using the axioms of independence and dependence atoms. 
 %For another application of team semantics in quantum information theory, see \cite{2014arXiv1409.5537H}. 
 Similarly, in the foundations of social choice theory, there are results such as Arrow's Theorem which  can also be formalized in the team semantics setting   \cite{2014arXiv1409.5537H,PacuitY16}.
 %%%%

For the applications it is important to understand the complexity theoretic aspects  of dependence logic and its variants. %Full dependence logic is  equivalent in expressive power with existential second-order logic ($\ESO$) \cite{vaananen07}. 
 During the past few years,  these aspects have been addressed in several studies. We will next briefly discuss some previous work. 
%\begin{itemize}
 The data complexity of inclusion logic  is sensitive to the choice between the two main variants of team semantics: under the so-called  lax semantics it is equivalent to positive greatest fixed point logic ($\PGFP$) and captures $\Ptime$ over finite (ordered) structures \cite{gallhella13}. Recently a  fragment of inclusion logic that captures $\textrm{NL}$ has been identified in \cite{HannulaH19}. The same article also exhibits surprisingly simple formulas of inclusion logic whose data complexity is $\nl$ and $\mathrm{P}$-complete (see equations  \eqref{HH1} and \eqref{HH2}).
 On the other hand, under the strict semantics, inclusion logic is equivalent to $\ESO$ (existential second order logic) and hence captures $\np$ \cite{galhankon13}.  %The  question whether there is a natural  fragment of dependence logic capturing $\Ptime$  was recently considered in \cite{DBLP:journals/corr/abs-1210-3321} %For $\ESO$  it is known that a so-called Horn fragment $\ESOHORN$ of  $\ESO$ captures 
% $\Ptime$ over  successor structures \cite{DBLP:journals/tcs/Gradel92}. 
%and 
% a fragment  $\SDHORN$  satisfying 
 %$\SDHORN = \ESOHORN=\Ptime$ over finite successor structures was identified.
 
 In \cite{durand11} the fragment of dependence logic allowing only sentences in which dependence atoms of arity at most $k$ may appear (atoms  $\dep(x_1,\ldots, x_n)$ satisfying $n\le k+1$) was shown to correspond to the $k$-ary fragment  $\ESOfarity{k}$ of ESO in which second-order quantification is restricted to at most $k$-ary functions and relations.  
 %Also,  the fragment $\dforall{k}$ in which at most $k$ variables are allowed to  be universally quantified was related 
 %$$\dforall{k} \le \ESOfvar{k}\le  \dforall{2k},$$
%to  a fragment   $\ESOfvar{k}$ of   ESO  consisting of  Skolem normal form sentences with  at most  $k$ universal first-order quantifiers.  % $\ESOfvar{k} =\NTIME_{\RAM}(n^k)$  \cite{grandjean04}.
Several  similar results have been obtained  also for independence and inclusion logic, e.g.,  in \cite{galhankon13,DBLP:journals/corr/HannulaK14,Ronnholm18,Ronnholm19}.

 The combined complexity of the model-checking problem of dependence logic, and many of its variants, was shown to be NEXPTIME-complete \cite{gradel12}. %Furthermore, for any variant of dependence logic whose atoms are PTIME-computable, the corresponding model checking problem is in NEXPTIME   
 On the other hand, the satisfiability problem for the two variable fragment of dependence logic (and many of its variants) was shown to be NEXPTIME-complete in \cite{KontinenKLV11, KontinenKV14}. Furthermore,  during the past few years, the complexity aspects of propositional and modal logics in team semantics have been also systematically studied (see \cite{Luck19,HannulaKVV18} and the references therein).

The starting point for the present work are the following results of \cite{kontinenj13} showing that the non-classical interpretation of disjunction in team semantics makes the model-checking of certain quantifier-free
formulas very complicated.  Define $\phi_1$ and $\phi_2$ as follows:
\begin{enumerate}
\item $\phi_1$ is the formula $\dep(x,y)\vee\dep(u,v)$, and
\item\label{phi2} $\phi_2$ is  the formula $ \dep(x,y)\vee\dep(u,v)\vee \dep(u,v)$.
\end{enumerate}
Surprisingly, the data complexity of the model-checking problem of $\phi_1$ and $\phi_2$   are already  NL-complete and NP-complete, respectively. In  \cite{kontinenj13} it was also shown that  model-checking for $\phi \vee \psi$ where $\phi$ and $\psi$ are $2$-\emph{coherent} quantifier-free formulas of dependence logic is always in $\nl$. A formula $\phi$ is  called $k$-coherent if, for all $\mA$ and $X$,  $\mA\models_X \phi$, if and only if,   
$\mA\models_Y \phi$ for all $Y\subseteq X$ such that $|Y|=k$. Note that the left-to-right implication is always true due to the downwards closure property of dependence logic formulas. The  downwards closure property also implies that, for dependence logic formulas,  the strict and the lax semantics are equivalent. For independence and inclusion logic formulas this is not the case.
%\arnaud{Should we mention that in all cases, strict is at least as hard as lax?}

In this article our goal is to shed light on  the tractability frontier of data complexity of dependence, independence, and inclusion logic  formulas under both strict and lax team semantics.  In order to state our results, we define  a new syntactic measure called the disjunction-width $\Disjunctiondepth{\phi}$  of a formula $\phi$. Our results show that,  for quantifier-free formulas $\phi$ of dependence logic, the data complexity of model-checking is in NL if $\Disjunctiondepth{\phi}\le 2$. Surprisingly, for independence logic the case of quantifier-free formulas turns out to be  more fine-grained. In particular, we exhibit  a quantifier-free formula  with $\Disjunctiondepth{\phi}\le 2$ whose data-complexity is NP-complete and also identify a more restricted fragment with data complexity in  $\nl$.
% This result is in a sense a explicit syntactically defined generalisation  of the result of  \cite{kontinenj13}  on $2$-coherent quantifier-free $\df$-formulas mentioned above. 
%We also show that for simple formulas $\phi$ with $\Disjunctiondepth{\phi}=3$  the problem becomes NP-complete. Note that  since  $\dep(x,y)$ is logically equivalent to $y \bot_x y$,  formula $\phi_2$ in \eqref{phi2} gives rise to an independence logic formula with these properties. 
%The task of  finding a formula with NP-complete model-checking problem  and $\Disjunctiondepth{\phi}=3$ that does not use \emph{conditional} independence atoms but only so-called \emph{pure} atoms  $y \bot x$ turns out to be more laborious (see Theorem \ref{pure3-disjunction}).
For quantified formulas,  the complexity 
%of model-checking 
is shown to be   NP-complete already with simple formulas constructed in terms of existential quantification and conjunction   in the empty non-logical vocabulary. 

For inclusion logic, we  show that model-checking can be reduced to the satisfiability problem of dual-Horn propositional formulas. While interesting in its own right, this also provides an alternative proof for the fact (see \cite{gallhella13}) that the  data complexity of  inclusion logic is in PTIME, and is also analogous  to the classical result of Gr\"adel on the Horn fragment of second-order logic \cite{DBLP:journals/tcs/Gradel92}.  We  also show that, under the strict semantics, the tractability frontier of model-checking of  (both quantifier-free and quantified) inclusion logic formulas becomes similar to that of  dependence and independence logic.

\section{Preliminaries}
In this section we briefly discuss the basic definitions and results needed in this article.

 \begin{definition}
Let $\mA$ be a structure with domain $A$, and $V=\{x_1,\ldots,x_k\}$ be a finite (possibly empty) set  of
variables.
\begin{itemize}
\item A \emph{team} $X$ of $\mA$ with domain $\Dom(X) = V$  is a finite set of assignments $s\colon V\rightarrow A$.
\item For a tuple $\tu x=( x_1,\ldots,x_n)$, where  $x_i\in V$,  $X(\tu x):=\{s(\tu x) : s \in X\}$ is the $n$-ary relation of $A$, where $s(\tu x):=( s(x_1),\ldots,s(x_n))$.
\item For $W\subseteq V$, $X \upharpoonright W$ denotes  the team obtained by restricting all assignments of $X$ to  $W$.
\item  The set of free variables of a formula $\phi$ is  defined as in first-order logic, taking into account that free variables may arise also from dependence, independence and inclusion atoms, and is denoted by $\Fr(\phi)$.
\end{itemize}
\end{definition}
%We are now ready to define team semantics. As now customary, we will restrict attention to formulas in negation normal form, and use the Lax semantics introduced in \cite{galliani12} that differs slightly from the semantics defined in \cite{vaananen07}.  
We will consider two variants of the semantics called the strict and
the original semantics given in \cite{vaananen07} is a combination of these variants (with the lax disjunction and the strict existential quantifier). For dependence logic formulas, the two variants of the semantics are easily seen to be equivalent, but for independence and inclusion logic this is not the case. We first define the lax  team semantics for first-order formulas in negation normal form. Below  $\mA \models_s \alpha$ refers to the satisfaction in first-order logic, and $s(m/x)$ is the assignment such that $s(m/x)(x)=m$, and   $s(m/x)(y)=s(y)$ for $y\neq x$. The power set of a set $A$ is denoted by $\mathcal{P}(A)$.

\begin{definition}
Let $\mA$ be a structure, $X$ be a team of $A$, and $\phi$ be a first-order formula such that $\Fr(\phi)\subseteq \Dom(X)$. %Then $\mA \models_X \phi$, if 
\begin{description}
\item[lit:] For a first-order literal $\alpha$, $\mA \models_X \alpha$ if and only if, for all $s \in X$, $\mA \models_s \alpha$.
\item[$\vee$:]  $\mA \models_X \psi \vee \theta$ if and only if,  there are $Y$ and $Z$ such that $Y \cup Z=X$,  $\mA \models_Y \psi$ and $\mA \models_Z \theta$.
\item[$\wedge$:] $\mA \models_X \psi \wedge \theta$ if and only if, $\mA \models_X \psi$ and $\mA \models_X \theta$.
\item[$\exists$:]  $\mA \models_X \exists x \psi$ if and only if, there exists a function $F : X \rightarrow \mathcal{P}(A)\setminus \{\emptyset\}$ such that $\mA \models_{X(F/x)} \psi$, where $X(F/x) =  \{s(m/x) : s \in X, m \in F(s)\}$.
\item[$\forall$:] $\mA \models_X \forall x \psi$ if and only if, $\mA \models_{X(A/x)} \psi$, where $X(A/x) = \{s(m/x) : s \in X, m \in A\}$.
\end{description}
A sentence $\phi$ is  \emph{true} in $\mA$ (abbreviated $\mA \models \phi$) if $\mA \models_{\{\emptyset\}} \phi$. Sentences $\phi$ and $\phi'$ are  \emph{equivalent}, $\phi \equiv \phi'$, if for all models $\mA$, $\mA \models \phi \Leftrightarrow \mA \models \phi'$.
\end{definition}
In the \emph{Strict Semantics}, the semantic rule for disjunction is modified by adding the requirement $Y\cap Z = \emptyset$, and  the clause for the existential quantifier  is replaced by 
\begin{description} 
\item $\mA \models_X \exists x \psi$ if and only if, there exists a function $H : X \rightarrow A$ such that $\mA \models_{X(H/x)} \psi$, where $X(H/x) = \{s(H(s)/x) : s \in X\}$. \end{description}
The meaning of first-order formulas is invariant under the choice between the strict and the lax semantics. 
First-order formulas satisfy what is known as the \emph{flatness} property:
%\begin{theorem}[Flatness]\label{flatness}Let $\mA$ be a structure and $X$ a team of $\mA$. Then for a  first- order formula $\phi$ the following are equivalent: 
%\begin{enumerate}
 $\mA \models_X \phi$, if and only if,  $\mA \models_s \phi$  for all $s \in X$.
Next we will give the semantic clauses for the new dependency atoms:
\begin{definition}
\begin{itemize}
\item Let $\tu x$ be a tuple of variables and let $y$ be another variable. Then $\dep(\tu x, y)$ is a \emph{dependence atom}, with the semantic rule
\begin{description}
\item $\mA \models_X \dep(\tu x, y)$ if and only if for all $s, s' \in X$, if $s(\tu x)=s'(\tu x)$, then $s(y)=s'(y)$;
\end{description}
\item Let $\tu x$, $\tu y$, and $\tu z$ be tuples of variables (not necessarily of the same length). Then $\indep{\tu x}{\tu y}{\tu z}$ is a \emph{conditional independence atom}, with the semantic rule
\begin{description}
\item $\mA \models_X {\tu x}\ \bot_{\tu y}\ {\tu z}$ if and only if for all $s, s' \in X$ such that $s(\tu y)=s'(\tu y)$, there exists an assignment $s'' \in X$ such that $s''(\tu x\tu y \tu z)=s(\tu x\tu y)s'(\tu z)$.
\end{description}
Furthermore, when $\tu z$ is empty, we write ${\tu x}\bot{\tu y}$ as a shorthand for $\indep{\tu x} {\tu z}  {\tu y}$, and  call it a \emph{pure independence atom};
\item Let $\tu x$ and $\tu y$ be two tuples of variables of the same length. Then $\tu x \subseteq \tu y$ %and $\tu x \mid \tu y$ are 
is an \emph{inclusion atom}, with the semantic rule %and \emph{exclusion} atoms, respectively, with the semantic rules
\begin{description}
\item $\mA \models_X \tu x \subseteq \tu y$ if and only if for all $s\in X$ there exists a $s' \in X$ such that $s'(\tu y)=s(\tu x)$. 
%\item[TS-exc:] $\M \models_X \tu x \mid \tu y$ if and only if $X(\tu x) \cap X(\tu y) = \emptyset$.
\end{description}
\end{itemize}
\end{definition}
%For a  collection $\mathcal C \subseteq \{=\!\!(\ldots), \bot_{\rm c}, \subseteq, \mid\}$ of atoms, we write $\FO(\mathcal C)$ (omitting the set parenthesis of $\mathcal{C}$) for the logic obtained by adding them to the syntax of first-order logic.
The formulas of dependence logic, $\df$, are obtained by extending the syntax of $\fo$ by dependence atoms. The semantics of $\df$-formulas is obtained by extending Definition 2 by the semantic rule defined above for dependence atoms.  Independence logic, $\indlogic$, and inclusion logic, $\FOinclusion$, are defined analogously using independence and inclusion atoms, respectively.
%The fragment of independence logic containing only pure independence atoms in denoted  $\indNRlogic$.  The aforementioned atoms are particular instance of a general notion of \emph{generalised dependence atom} \cite{DBLP:journals/corr/Kuusisto13}. The semantics of a generalised dependence atom $A_Q$ is determined (essentially) by a class  $Q$ of structures and teams over which the atomic formula $A_Q(\tu x_1,\ldots, \tu x_n)$  is satisfied (see \cite{DBLP:journals/corr/Kuusisto13} for details).  

It is easy to see that the flatness property is lost immediately when $\fo$ is extended by any of the atoms defined above. On the other hand,  it is straightforward to check that all $\df$-formulas satisfy the following  strong \emph{downwards closure} property:  if $\mA \models_X \phi$ and   $Y\subseteq X$, then $\mA \models_Y \phi$.
%\begin{proposition}[Downwards closure]
%\label{thm:dc}
%Let  $\phi$ be a $\df$-formula. Then for all structures  $\mA$  and teams  $X$, if $\mA \models_X \phi$ and    $Y\subseteq X$, then $\mA \models_Y \phi$.
%\end{proposition}
Another basic property of logics in team semantics is called \emph{locality}:
% $\mA \models_X \phi$, if and only if,    $\mA \models_{X \upharpoonright \Fr(\phi)} \phi$.
\begin{proposition}[Locality]
Let  $\phi$ be a formula of any of the logics  $\df$,  $\indlogic$ or $\FOinclusion$. Then, under the lax semantics, for all structures  $\mA$  and teams  $X$:
 $$\mA \models_X \phi \mbox{ iff }    \mA \models_{X \upharpoonright \Fr(\phi)} \phi.$$
\end{proposition}
Under the strict semantics locality  holds only for dependence logic formulas (see \cite{galliani12} for details).

In this article we study the  data complexity  of model-checking  of dependence, independence, and inclusion logic formulas. In other words, for a fixed formula $\phi$ of one of the aforementioned logics, we study the complexity of the following model-checking  problem:  given a finite structure  $\mA$ and a team $X$,  decide whether $\mA\models_X\phi$. 
Note that when we are working with the lax semantics, we may assume without loss of generality that the  domain of $X$ is exactly $\Fr(\phi)$. \emph{If not explicitly mentioned otherwise, all  results are valid under both strict and lax semantics}.
We assume that the reader is familiar with the basics of computational 
complexity theory such as NP-completeness and the log-space bounded classes $\mathrm{L}$ and $\mathrm{NL}$.

\subsection{Complexity classes and satisfiability problems}

In the paper, we will make use of some well-known computational problems whose complexity is recalled here (unless specified all proofs can be found in~\cite{GareyJ1979}). We suppose the reader familiar with basic complexity classes such as $\logspace$ (logarithmic space), $\nl$ (non deterministic logarithmic space), $\Ptime$ (polynomial time), $\np$ (non deterministic polynomial time).

Let $k\in \mathbb{N}$, a propositional formula is in $k$-\pb{cnf} if it is in conjunctive normal form with clauses of length at most $k$. It is positive, if it equivalent to a formula without negation. 

It is well known that, $k$-$\pb{sat}$ the satisfiability problem for $k$-\pb{cnf} formulas is $\np$-complete for $k\geq 3$ and $\nl$-complete for $k=2$. \pb{horn-sat} (resp. \pb{dual-horn-sat}) is the satisfiability problem of \pb{cnf} formulas with at most one positive (resp. negative) literal per clause. This problem is known to be $\Ptime$-problem.

Given a positive $3$-\pb{cnf} formula, deciding whether $\phi$ is satisfiable such that exactly one variable is set to true in each clause is also known to be $\np$-complete. This problem is called \pb{1-in-3-sat}.

\section{Dependence and independence logics}%
%In the rest of this article, due to lack of space,  the full proofs of selected results can only be found in the Appendix.
%Due to lack of space,  the full proofs of some results can only be found in the Appendix.
In this section we consider the complexity of model-checking for quantifier-free and quantified  formulas of dependence and independence logic. 

\subsection{The case of quantifier-free formulas}
In this section we consider the complexity of model-checking for quantifier-free formulas of dependence and independence logic. For dependence logic the problem has  already been essentially settled in \cite{kontinenj13}.  The following theorems delineate a clear barrier between tractability and intractability for quantifier-free dependence logic formulas.
%known about the data complexity of the model cheking problem for formulas in $\df$. 
\begin{theorem}[\cite{kontinenj13}]\label{JarmoKontinenNL} The model checking problem for the formula 
$$\dep(x,y)\vee \dep(z,v)$$ 
is $\nl$-complete. More generally, the model-checking for $\phi \vee \psi$ where $\phi$ and $\psi$ are $2$-coherent quantifier-free formulas of $\df$ is always in $\nl$.
\end{theorem}

When  two disjunctions can be used, the model checking problem becomes intractable as shown by the following results.

\begin{theorem}[\cite{kontinenj13}]\label{JarmoKontinenNP}
The model checking problem for the formula 
$$\dep(x,y)\vee \dep(z,v) \vee \dep(z,v)$$
 is $\np$-complete.
\end{theorem}

%We will show below that the formula above is in some sense optimal. We start with the following result which essentially shows that two independence atoms  are not enough to express $\np$-complete problems. 
%\arnaud{The result below will supersede Proposition \ref{2 disjunction Indep} and \ref{2 disjunction Indep bis}}

In order to give a syntactic generalization of  Theorem \ref{JarmoKontinenNL}, we define next the disjunction-width of a formula. % The type of a formula $\phi$, denoted $\type{\phi}$, is 1 if $\phi$ contain at least one "dependence" atom i.e. is not purely first-order . 

\begin{definition} Let $\sigma$ be a relational signature. The disjunction-width of a $\sigma$-formula $\phi$, denoted $\Disjunctiondepth{\phi}$, is defined as follows:
\[
\Disjunctiondepth{\phi}=
\left\{\begin{array}{l}
1  \mbox{ if } \phi  \mbox{ is } \indep{\tu y}{\tu x }{\tu z} \mbox{ or }  \dep(\tu x, y) \mbox{ or } \tu x\subseteq \tu y\\
0   \mbox{ if } \phi \mbox{ is } R(\tu x)  \mbox{ or }   \neg R(\tu x),    \mbox{ for } R\in \sg\cup\{=\}\\
\max (\Disjunctiondepth{\phi_1}, \Disjunctiondepth{\phi_2}) \mbox{ if } \phi \mbox{ is } \phi_1\wedge\phi_2 \\
\Disjunctiondepth{\phi_1} + \Disjunctiondepth{\phi_2} \mbox{ if } \phi \mbox{ is } \phi_1\vee\phi_2 \\
\Disjunctiondepth{\phi_1}  \mbox{ if } \phi \mbox{ is } \exists x \phi_1 \mbox{ or } \forall x \phi_1.
\end{array}\right.
\]
\end{definition}

The next theorem  is a syntactically defined analogue of Theorem  \ref{JarmoKontinenNL}.
\begin{proposition}\label{disjunction depth 2}
The data complexity of  model-checking of quantifier-free $\df$-formulas  $\phi$ with $\Disjunctiondepth{\phi} \leq 2$ is in  $\nl$.
\end{proposition}

\begin{proof} We will first show that a formula  $\phi$ with $\Disjunctiondepth{\phi} =1$ is $2$-coherent.
 This follows by induction using the following facts \cite{kontinenj13}: 
 \begin{itemize}
\item dependence atoms are $2$-coherent, and first-order formulas are $1$-coherent,
\item if $\psi$ is $k$-coherent, then $\psi\vee \phi$ is also $k$-coherent assuming  $\phi$ is first-order,
\item  if   $\psi$ is $k$-coherent and  $\phi$ is $k\le j$-coherent, then  $\psi\wedge \phi$ is $j$-coherent.
\end{itemize}
It is also straightforward to check that the data complexity of a formula  $\phi$ with $\Disjunctiondepth{\phi} =1$ is in L (the formula $\phi$ can be  expressed in FO assuming the team $X$ with domain $\tu x=\Fr(\phi)$ is represented by the $n$-ary relation $X(\tu x)$). We will complete the  proof using induction on $\phi$ with  
 $\Disjunctiondepth{\phi} = 2$.  Suppose that  $\phi=\psi\vee \eta$, where  $\Disjunctiondepth{\psi}= \Disjunctiondepth{\eta} = 1$. Then the claim follows by Theorem  \ref{JarmoKontinenNL}. The case $\phi=\psi\wedge \eta$ is also clear. Suppose finally  that $\phi=\psi\vee \eta$, where $\eta$ is first-order. Note that by downward closure and flatness
 $$\mA\models_X \phi \Leftrightarrow  \mA\models_{X'}\psi,$$
where $X'=\{ s\in X \ |\ \mA\not \models_s \eta \}$. Now since $X'$ can be computed in L,  the model-checking problem of $\phi$ can be decided in NL by the induction assumption for $\psi$. %Finally assume that   $\eta=\psi\vee \phi$, where  $\Disjunctiondepth{\psi}= \Disjunctiondepth{\phi} = 1$. A simple induction shows that the formulas $\psi$ and  $\phi$ 2-coherent using the facts that  
\end{proof}

In the rest of this section we examine potential  analogues of Theorems \ref{JarmoKontinenNL} and \ref{JarmoKontinenNP} for  independence logic. It is well-known that the dependence atom $\dep(\tu x,y)$ is logically equivalent to the independence atom $ \indep{y}{\tu x}{y}$. Hence, the following is immediate from Theorem~\ref{JarmoKontinenNP}.
%an immediate corollary of Theorem~\ref{JarmoKontinenNP}.

\begin{corollary}\label{3 disjunction relativized independance}
The model checking problem for the formula 
\[\indep{y}{x}{y}\vee \indep{v}{z}{v} \vee \indep{v}{z}{v}\]
 is $\np$-complete.
\end{corollary}

For independence logic, the situation is not as clear as for dependence logic concerning tractability.
 In the following we will exhibit a fragment of independence logic whose data complexity is in $\nl$ and which is in some sense the maximal such fragment. 
 
 %%%%%%%%
%%%%%%%%
%%%%%%%% JUHA: this does not yet work:

 %In \cite{HannulaH19} it was recently shown that the model checking problem for the fixed formula
 %$$x\subseteq y \vee u=v $$
%is already $\nl$-complete. Using this result it is straight forward to show the analoous result for independence atoms.

%\begin{proposition}\label{indepNL}
%The model checking problem for the formula 
%\[\indep{x}{}{y}\vee u=v}\]
 %is $\nl$-complete.
%\end{proposition}
%\begin{proof} We first show how $x\subset y$ can be easily expressed by a single  independence atom. The idea for the translation is adapted from \cite{galliani12}, where inclusion atoms are shown to be expressible in indepedence logic. Let  $X$ be a team. Define  $X^*$ with domain $\Dom(X)\cup \{z,b\}$ ($z$ and $b$ fresh variables) as follows:
%$$X^*:=\{ s(s(x)/z)(0/b) \ |\ s\in X  \}\cup \{ s(s(y)/z)(0/b), \  s(s(y)/z)(1/b) \ |\ s\in X  \}.  $$   
%It is now easy to check that $X\models x\subseteq y$ if and only if $X^*\models \indep {z}{}{b}$. Furthermore, by flattness of the formula $u=v$, it is immediate that $X\models x\subseteq y \vee u=v $ if and only if 
%$X^*\models \indep{x}{}{y}\vee u=v $.

%\end{proof}

\newcommand{\bcindepfo}{\mathsf{BC}(\bot,\fo)}
\newcommand{\CCm}{\mathfrak{C}^{-}}
\newcommand{\CCp}{\mathfrak{C}^{+}}

\begin{definition} The Boolean closure of an independence atom by first-order formulas, 
denoted $\bcindepfo$, is defined as follows:

\begin{itemize} 
\item Any independence atom $\indep{\tu x}{\tu y}{\tu z}$ is in $\bcindepfo$.

\item If $\phi\in \bcindepfo$, then for any formula $\phi\in \fo$, $\phi \wedge \phi$ and $\phi\vee \phi$ are  in $\bcindepfo$. 
\end{itemize}
\end{definition}

Let $\phi\in \bcindepfo$. Up to permutation of disjunction and conjunction, $\phi$ can be put into the following normalized form:

\begin{equation}\label{normalized}
\phi\equiv ((\dots ((\indep {\tu x}{\tu z}{\tu y} \wedge \phi_1) \vee \psi_1)\wedge \dots ) \wedge \phi_k) \vee \psi_k).
\end{equation}

For a structure $\mA$, $\CCp:=\bigcap_{i=1}^k \phi_i(\mA)$ and $\CCm:=\bigcup_{i=1}^k \psi_i(\mA)$, where 
 $\phi_i(\mA):=\{s \ |\ \mA\models_s \phi_i \}$.

The next lemma gives a combinatorial criterion for the satisfaction of $\bcindepfo$ formulas. 

\begin{lemma}\label{bcindepfo}
For all $\phi\in\bcindepfo$ of the form~(\ref{normalized}), structures $\mA$,  and teams $X$  the following are equivalent:
\begin{description}
\item [(A)]$\mA\models_X\phi$.
\item [(B)] (1) and (2) below hold: 
\begin{description}

\item[(1)] For all $s\in X$ either $s\in [((..(\psi_1\wedge \phi_2) \vee \dots ) \wedge \phi_k) \vee \psi_k](\mA)$ or $s\in\CCp\setminus\CCm$. 

\item[(2)]\label{c2}  For all $s_1,s_2\in X$ such that $s_1,s_2\in \CCp\setminus \CCm$, and $s_1(\tu z)=s_2(\tu z)$,   there exists $s_3\in  X\cap \CCp$,
%\juha{The original version of condition 2 stated that $s_3\in X$ but by Jouko's proof we need to require that $s_3\in  \CCp\setminus \CCm$ (unless I have misunderstood something)} \arnaud{I agree} 
such that: $s_3(\tu z)=s_1(\tu z), s_3(\tu x)=s_1(\tu x) \mbox{ and } s_3(\tu y)=s_2(\tu y)$. 

%In the rest of the paper, assignments $s_1,s_2$ as in the second item  and $s_3$ is called a \textit{witness} of $s_1,s_2$ (for formula $\phi$).

\end{description}

\end{description}Assignments  $s_1,s_2$ as in (2) will be called \textit{compatible} for formula $\phi$ and team $X$. Furthermore, $s_3$ as in (2) will be called a \textit{witness} of $s_1,s_2$ for formula $\phi$.
\end{lemma}

\begin{proof} Note that (A) is equivalent to the existence of subteams $X'$ and $Y_i$, $1\le i\le k$, of $X$ such that 

\begin{enumerate}
\item[(i)$_{k}$] $X=X'\cup Y_1\cup\ldots\cup Y_k,$
\item[(ii)$_{k}$] $\mA\models_{Y_i}\psi_i,$ for $1\le i\le k$,
\item[(iii)$_{k}$] $\mA\models_{X'\cup \bigcup_{1\le j <i} Y_{j}}\phi_{i},$ for $1\le i\le k$
\item [(iv)$_{k}$] $\ma\models_{X'}  \indep {\tu x}{\tu z}{\tu y}$.
\end{enumerate}
This can be proved by induction on $k$ as follows: The claim is clearly true for $k=1$. Suppose it holds for $k-1$, where $k>1$. We know $\phi$ is equivalent to $$ (((\dots ((\indep {\tu x}{\tu z}{\tu y} \wedge \phi_1) \vee \psi_1)\wedge \dots ) \wedge \phi_k) \vee \psi_k.
$$ Hence $\ma\models_X\phi$ if and only if $X=Z\cup Y_k$ such that $$\ma\models_Z (((\dots ((\indep {\tu x}{\tu z}{\tu y} \wedge \phi_1) \vee \psi_1)\wedge \dots ) \wedge \phi_{k-1}) \vee \psi_{k-1})\wedge\phi_k
$$ and $\ma\models_{Y_k}\psi_k$. In particular, $\ma\models_Z\phi_k$ and $$\ma\models_Z ((\dots ((\indep {\tu x}{\tu z}{\tu y} \wedge \phi_1) \vee \psi_1)\wedge \dots ) \wedge \phi_{k-1}) \vee \psi_{k-1}.$$
By Induction Hypothesis there are subteams $X'$ and $Y_i$, $1\le i< k$, of $Z$ such that (i)$_{k-1}$-(iv)$_{k-1}$ hold with $X$ replaced by $Z$.
Now the sequence $X',Y_1,\ldots,Y_k$ satisfies (i)$_{k}$-(iv)$_{k}$ and is therefore as required for the induction claim. The converse is similar.

(A) implies (B): To prove (1), suppose $s\in X$. Thus $s\in X'\cup Y_1\cup\ldots\cup Y_k$, where $X',Y_1,\ldots, Y_k$ are as defined earlier. Suppose  
$$s\notin [((..(\psi_1\wedge \phi_2) \vee \dots ) \wedge \phi_k) \vee \psi_k](\mA).
$$
 Then $\ma\not\models_s\psi_k$, whence $s\notin Y_k$. Hence $s\in X'\cup Y_1\cup\ldots \cup Y_{k-1}$, whence $\ma\models_s\phi_k$, and therefore  $$s\notin [((..(\psi_1\wedge \phi_2) \vee \dots ) \wedge \phi_{k-1}) \vee \psi_{k-1}](\mA).
$$ Eventually, step by step, we verify $s\in \CCp\setminus\CCm$ and hence
 (1) is proved. 
 
 To prove (2), suppose $s_1,s_2\in X\cap(\CCp\setminus \CCm)$ such that $s_1(\tu z)=s_2(\tu z)$. By (i)$_{k}$-(iii)$_{k}$, $s_1,s_2\in X'$. By (iv)$_{k}$ there is a witness as required.
 
 (B) implies (A): Let us denote by $X' \subseteq (X\cap \CCp )$ a minimal superset of   $X\cap(\CCp\setminus \CCm)$ satisfying   $\indep {\tu x}{\tu z}{\tu y}$. Such a set $X'$ exists by the assumption (2). Then, for $1\le i\le k$, define
 $$Y_i=(X\setminus X') \cap[\psi_i\wedge \bigwedge_{i<j\le k}\phi_{j}](\mA).$$
 Suppose $s\in X$ but $s\notin (\CCp\setminus \CCm)\subseteq X'$. By (1), $$s\in[((..(\psi_1\wedge \phi_2) \vee \dots ) \wedge \phi_k) \vee \psi_k](\mA).$$ If $s\notin Y_k$, then  
 $$s\in[((..(\psi_1\wedge \phi_2) \vee \dots \vee\psi_{k-1}) \wedge \phi_k)](\mA).$$ Hence $\ma\models_s \phi_k$.
 If $s\notin Y_{k-1}$, then 
 $$s\in[((..(\psi_1\wedge \phi_2) \vee \dots \vee\psi_{k-2}) \wedge \phi_{k-1})](\mA).$$ Hence 
 $\ma\models_s\phi_{k-1}$. Continuing this way we see that if $s\notin Y_1\cup\ldots Y_k$, then $s\in \CCp\setminus\CCm$, contrary to the assumption  $s\notin \CCp\setminus \CCm$. We have proved (i)$_{k}$. Clearly (ii)$_{k}$ and  (iii)$_{k}$ hold. Finally, (iv)$_{k}$ holds by the definition of $X'$.
\end{proof}

%The first item is true by exhaustive case distinction. The second one comes from the fact that if a tuple $s$ satisfies  $s\in\CCp$ and  $s\not\in\CCm$  then it is forced to be in the sub-team satisfying $\indep{\tu x}{\tu z}{\tu y}$.  

\begin{theorem}\label{BC(top,FO)} The data complexity of any $\phi\in\bcindepfo$ is in $\mathrm{L}$. This is true both in lax and in strict semantics.\end{theorem}
\begin{proof}
It is well known that checking whether a tuple $s$ belongs to the query result $\phi(\mA)$ of a first-order formula can be done in logarithmic space~\cite{Immerman99}.  Therefore, by the characterization of Lemma \ref{bcindepfo}, deciding whether $\mA\models_X\phi$ is also in $\mathrm{L}$.
For the strict semantics, it suffices to note that the proof of Lemma \ref{bcindepfo}  goes through  also under the strict semantics. We can use the fact that the formulas $\phi_i$ have the downwards closure property to force the  subteams $X', Y_1,..,Y_k$ in the decomposition of $X$ to be pairwise disjoint.
\end{proof}

Interestingly the analogue of Lemma \ref{bcindepfo} does not hold for inclusion atoms; It was recently shown in \cite{HannulaH19} that the data complexity of the formula
\begin{equation}\label{HH1}
 x\subseteq y \vee u=v 
\end{equation}
is already $\nl$-complete and for 
\begin{equation}\label{HH2}
(x\subseteq z \wedge y\subseteq z ) \vee u=v 
\end{equation}
the problem becomes $\mathrm{P}$-complete. These results hold under both strict and lax semantics as the other disjunct in both formulas is first-order and satisfies the downwards closure property.
Furthermore, in the last section of this article we construct a quantifier free inclusion logic formula with $\np$-complete data complexity under the strict semantics.
\begin{theorem}\label{thm13}
The data complexity of  formulae of the form $\phi_1\vee \phi_2$  with $\phi_1,\phi_2\in \bcindepfo$ is in $\nl$ under the lax semantics.
\end{theorem}

\begin{proof} 
The proof is given by a log-space reduction to the $\nl$-complete problem $2$-\pb{sat}.
Given a structure $\mA$ and a  team $X$ we  construct a $2$-\pb{cnf} propositional formula $\Phi$ such that:
\begin{equation}\label{indep-translation}
 \mA \models_X \phi_1\vee \phi_2 \iff \Phi \mbox{ is satisfiable.} 
\end{equation}

 Recall that if a team $X$ is such that  $\mA \models_X \phi_1\vee \phi_2$ then, there exists $Y, Z\subseteq X$ such that $Y\cup Z=X$ and $\mA \models_Y \phi_1\mbox{ and } \mA \models_Z \phi_2$.
 % \[  \mA \models_Y \phi_1\mbox{ and } \mA \models_Z \phi_2.\]
For each assignment $s\in X$, we introduce two Boolean variables $Y[s]$ and $Z[s]$. Our Boolean formula $\Phi$ will be defined below with these  $2|X|$ variables the set of which is denoted by $\Var(\Phi)$. It will express that the set of assignments must split into $Y$ and $Z$ but also make sure that incompatible (see Lemma~\ref{bcindepfo}) assignments do not appear in the same subteam. 
 
For each pair $s_i,s_j$ that are incompatible for $\phi_1$ on team $X$, one adds the $2$-clause: $\neg Y[s_i] \vee \neg Y[s_j]$. The conjunction of these clauses is denoted by $C_Y$.
Similarly, for each pair $s_i,s_j$ that are incompatible for $\phi_2$ on team $X$, one adds the clause: $\neg Z[s_i] \vee \neg Z[s_j]$ and call $C_Z$  the conjunction of these clauses.

%\[
%\neg Y[s_i] \vee \neg Y[s_j]. 
%\]
%
%The conjunction of these clauses is denoted by $C_Y$.
%Similarly, for each pair $s_i,s_j$ that are incompatible for $\phi_2$ on team $X$, one adds: 
%
%\[
%\neg Z[s_i] \vee \neg Z[s_j], 
%\]
%
%\noindent and analogously $C_Z$  denotes the conjunction of these clauses.
Finally, the construction of $\phi$ is completed by adding the following conjunctions:

\[ C^0:= \bigwedge \{Y[s] \vee Z[s]:{s\in X}\} \]
\[ C^1:= \bigwedge\{ \neg Y[s]:{s \textrm{ fails (B,1) of Lemma \ref{bcindepfo}}  \textrm{ for } \phi_1} \}   \]
\[ C^2:= \bigwedge\{ \neg Z[s]:{s \textrm{ fails (B,1) of Lemma \ref{bcindepfo}}  \textrm{ for } \phi_2}\}   \]

It is not hard to see  that the  formula
 $$\Phi\equiv \bigwedge _{0\le i \le 2}C^i\wedge C_Y\wedge C_Z$$
  can be constructed deterministically in log-space.
 It remains to show that the equivalence \eqref{indep-translation} holds.
 
Assume that the left-hand side of the equivalence holds. Then, there exists $Y, Z\subseteq X$ such that $Y\cup Z=X$, $\mA \models_Y \phi_1$ and  $\mA \models_Z \phi_2$.
%\begin{equation}\label{disjunction} 
%  \mA \models_Y \phi_1\mbox{ and } \mA \models_Z \phi_2.
%  \end{equation}
% 
 We construct a propositional assignment $I:\Var (\Phi)\rightarrow \{0,1\}$ as follows. For all $s\in Y$, we set $I(Y[s])=1$ and for all $s\in Z$, we set similarly $I(Z[s])=1$. It is now immediate that the all of the clauses in  $\bigwedge _{0\le i \le 2}C^i$ are satisfied by $I$.
 
 Let us consider a clause  $ \neg Y[s_i] \vee \neg Y[s_j]$ for an incompatible pair $s_i,s_j$. Then, $I(Y[s_i])=0$ or $I(Y[s_j])=0$ must hold. For a contradiction, suppose that $I(Y[s_i])=I(Y[s_j])=1$. Then since $\mA \models_Y \phi_1$ holds, by construction $s_i$ and $s_j$ must be compatible for $\phi_1$. Hence we get a contradiction and may conclude that $I$ satisfies $ \neg Y[s_i] \vee \neg Y[s_j]$. The situation is similar for each clause $\neg Z[s_i] \vee \neg Z[s_j]$. %Finally since $X=Y\cup Z$, $I$ also satisfies $C$.
 
Let us then assume that  $\Phi$ is satisfiable, and let $I:\Var (\Phi)\rightarrow \{0,1\}$ be a satisfying assignment for $\Phi$. Since $I\models C^0$, we get that $I(Y[s])=1$ or $I(Z[s])=1$ for all  $s\in X$. Let 
   $$ X_{Y}=\{s : I(Y[s])=1\} \mbox{ and } X_{Z}=\{s : I(Z[s])=1\}.  $$
 Now $X_Y\cup X_Z=X$ and clauses $C^1$  ($C^2$) ensure that each $s\in X_Y$ ($s\in X_Z$) satisfies condition (B,1) of Lemma   \ref{bcindepfo}.
 
 % \arnaud{This part of the proof should be checked carefully. I am not sure it is the best way to write it.}
 We will next show how the  sets $X_Y$ and $X_Z$ can be extended to sets $Y$ and $Z$ satisfying also condition (B,2) of Lemma   \ref{bcindepfo} and consequently $\mA \models_Y \phi_1$ and  $\mA \models_Z \phi_2$.
 Note first that, since $I$ satisfies $\Phi$, for all $s_1,s_2\in X_{Y}$, $\Phi$ cannot have a clause of the form $\neg Y[s_1]\vee \neg Y[s_2]$, and hence $s_1,s_2$ are compatible for $\phi_1$. Analogously we see that all  $s_1,s_2\in X_{Z}$ are compatible for  $\phi_2$.
We will define the sets  $Y$ and $Z$ incrementally by first initializing them to  $X_Y$ and $X_Z$, respectively. Note that even if $X_Y\cup X_Z=X$, no decision has been made regarding the membership of   assignments $s$ in $Y$ (resp. $Z$) such that $I(Y[s])=0$ (resp. $I(Z[s])=0$).  Let us first consider $Y$. Until no changes occur,  we consider all pairs $s_1,s_2\in Y\cap (\CCp\setminus\CCm)$ (where the sets $\CCp$ and $\CCm$ are taken with respect to $\phi_1$) such that  $s_1(\tu z)= s_2(\tu z)$  and add into $Y$ (if they are not already in) all tuples $s_3$ such that $s_3$ is a witness for the pair $(s_1,s_2)$ regarding formula $\phi_1$. Since by construction $s_1,s_2$ are compatible then at least one such  $s_3$ exists (but may be initially outside of $Y$). We prove below that this strategy is safe.
 %\begin{itemize}
 %\item either $s_1[\tu x_1]\neq s_2[\tu x_1]$. In this case, we simply add $s_1$ and $s_2$ if they are not in already.
  %\item or $s_1[\tu x_1]= s_2[\tu x_1]$. 
% In this case,
 %there exists $s_3,s_4\in X$ such that $\star(s_1,s_2,s_3,s_4)$ is true.
 % \[  s_1[\tu u_1]=s_3[\tu u_1], s_2[\tu v_1]=s_3[\tu v_1] \mbox{ and }  s_2[\tu u_1]=s_4[\tu u_1], s_1[\tu v_1]=s_4[\tu v_1].\]
First of all, it is  easily seen that any pair among $\{s_1,s_2,s_3\}$ is compatible for $\phi_1$. Therefore, it remains to show that the new assignments $s_3$  are compatible with every other element $s$ added to $Y$ so far. Suppose this is not the case and that there exists $s\in (Y \cap (\CCp\setminus\CCm)) \setminus \{s_1,s_2\}$ such that $s_3$ and $s$ are incompatible for $\phi_1$. Note that  for incompatibility we  must have $s_3(\tu z)=s(\tu z)$. Since $s_1,s_2$, and $s$ are in $Y$ they are all pairwise compatible. Hence, there exists $t_1$ such that $t_1$ is a witness for the pair $(s_1,s)$.
Then,  $t_1(\tu x)=s_1(x)=s_3(\tu x)$, and  $t_1(\tu y)=s(\tu y)$. %Analogously, $t_2(\tu x)=s(\tu x)$ and $t_2(\tu y)=s_2(\tu y)=s_3(\tu y)$. 
 %But we also know that $s_2(\tu u_1)= s_4(\tu u_1)$ and $s_1(\tu v_1)= s_4(\tu v_1)$. 
 Consequently, $t_1$ is also a witness for $s_3,s$ hence, $s_3$ and $s$ are compatible which is a contradiction. Therefore,  the assignment $s_3$  can be safely added to $Y$. The set  $Z$ is defined analogously. %Note that $Y$ and $Z$ will not be necessarily disjoint.
By the construction,  the sets $Y$ and $Z$ satisfy conditions (B,1) and (B,2) of Lemma   \ref{bcindepfo} for $\phi_1$ and $\phi_2$, respectively,  and hence it holds that  $\mA \models_Y \phi_1$  and  $\mA\models_Z \phi_2$.

 %\[  \mA \models_Y \phi_1  \mbox{ and } \mA\models_Z \phi_2.\]
  
%\end{itemize}

%  Since the assignment $I$ also satisfy the subformula $C$ then, $X=Y\cup Z$.
  \end{proof}

It is worth noting that the last step of the proof of the above theorem does not seem to  work  under strict semantics. It is an open question whether Theorem \ref{thm13} holds under the strict semantics.

We will next show that a slight relaxation on the form of the formula immediately yields intractability 
of model-checking.
\begin{theorem}
There exists a formula $\phi_1\vee \phi_2$ the model-checking problem of which is $\np$-complete and such that:

\begin{itemize}
    \item $\phi_1\in\bcindepfo$ and
    \item $\phi_2$ is the conjunction of \textit{two} independence atoms.
\end{itemize} 
\end{theorem}

\begin{proof}
Define $\phi_1 := w\neq 1 \wedge x \perp_{t} y$, $\phi_2 := c_1\perp_c c_2 \wedge x \perp_z y $. We will reduce  $3$-SAT to the  model-checking problem  of $\phi_1\vee\phi_2$. %The data complexity of the formula:
%\[\phi \equiv \left(w\neq 1 \wedge x \perp_{t} y \right) \vee \left(c_1\perp_c c_2 \wedge x \perp_z y\right)\]
%is NP-complete by a reduction to $3$-SAT.
 Let $\Phi = \bigwedge_{i=1}^n C_i$ be  a $3$-SAT instance. 
 Each $C_i= p_{i_1}$ is of the form $p_{i_2} \vee p_{i_3}$ % is a disjunction of $3$ literals
  with  $p_{i_1}, p_{i_2}, p_{i_3} \in \{\, v_1,\hdots, v_m,\neg v_1,\hdots, \neg v_m\,\}$. To this instance we associate a structure $\mA$ of the empty vocabulary and a team $X$ on the variables $w,c,c_1,c_2,z,x,y,t$. The the universe of the structure $\mA$ is composed of $\{1,\ldots,m\}$, $m$ new elements $a_1,\hdots,a_m$ and of $\{\, v_1,\hdots, v_m,\neg v_1,\hdots, \neg v_m\,\}\cup\{\,0,1\,\}$. For each clause $C_i$ we add  in $X$ the $6$ assignments displayed on the left below,  and for each variable $v_i$,  we add to $X$ the $2$ assignments on the right:
\[
\begin{array}{lcr}\begin{array}{|c|c|c|c|c|c|c|c|}
\hline
w&c&c_1&c_2&z&x&y&t\\
\hline \hline
0&i&1&1&i_1&p_{i_1}&p_{i_1}&a_{6i+1}\\
0&i&1&1&i_2&p_{i_2}&p_{i_2}&a_{6i+2}\\
0&i&1&1&i_3&p_{i_3}&p_{i_3}&a_{6i+3}\\
\hline
1&i&0&0&0&0&0&a_{6i+4}\\
1&i&1&0&0&0&0&a_{6i+4}\\
1&i&0&1&0&0&0&a_{6i+4}\\
\hline
\end{array}& &
\begin{array}{|c|c|c|c|c|c|c|c|}
\hline
w&c&c_1&c_2&z&x&y&t\\
\hline \hline
0&0&0&0&i&v_i&v_i&a_{6(n+1)+i}\\
0&0&0&0&i&\neg v_i&\neg v_i&a_{6(n+1)+i}\\
\hline
\end{array}
\end{array}\]
We will next show that $\Phi$ is satisfiable if and only if $\mA \models_X \phi$.

%\begin{itemize}
%	\item[$\Rightarrow$] 

Suppose there is an assignment  $I:\{\,v_1,\hdots,v_m\,\} \rightarrow \{0,1\}$ that evaluates  $\Phi$ to true, i.e., at least one literal in each clause is evaluated to $1$. We have to split $X$ into two sub-teams $X = Y \cup Z$ such that $\mA \models_Y \left(w\neq 1 \wedge x \perp_{t} y \right)$ and $\mA \models_Z \left(c_1\perp_c c_2 \wedge x \perp_z y\right)$. We must put every assignment $s \in X$ such that $s(w) = 1$ in $Z$. There are exactly three such assignments per clause. We put in $Z$ every assignment $s$ such that $s(x) = v_i$ if $I(v_i) = 1$, and $s(x) = \neg v_i$ if $I(v_i) = 0$. The other assignments are put into $Y$.
	
	For each clause $C_i$, one literal $p_{i_1},p_{i_2},p_{i_3}$ is assigned to $1$ by $I$. Then there is at least one assignment $s(c,c_1,c_2) = (i,1,1)$ in $Z$. In $Z$, the  assignments mapping $c$ to $i$ map $(c,c_1,c_2)$ to $(i,1,1), (i,1,0),(i,0,1)$ or $(i,0,0)$. Thus $\mA \models_Z c_1\perp_c c_2$. 
	
	If $s_1,s_2 \in Z$  are such that $s_1(z)=s_2(z) = i$, then $s_1(x)$ (analogously $s_2(x)$) is $v_{i}$ if $I(v_i) = 1$, $\neg v_i$ otherwise.  Therefore, $s_1(x) = s_2(x) = s_1(y) = s_2(y)$, and hence $\mA \models_Z x \perp_z y$  holds.
	
	As for $Y$, it is immediate that $\mA \models_Y w\neq 1$. The only pair of assignments $s_1,s_2$ in $X$ such that $s_1(t)=s_2(t)$ are $(0,0,0,0,i,v_i,v_i,a_k)$ and $(0,0,0,0,i,\neg v_i,\neg v_i,a_k)$, for some $k$. Only one of them is in $Y$ ($s_1$ if $I(v_i) = 1$, $s_2$ otherwise). Thus $\mA \models_Y x \perp_t y$.
	
%	\item[$\Leftarrow$] 

Suppose then that $X = Y \cup Z$ such that $\mA \models_Y \left(w\neq 1 \wedge x \perp_{t} y \right)$ and $\mA \models_Z \left(c_1\perp_c c_2 \wedge x \perp_z y\right)$. Define an assignment  $I$ of the variables of $\Phi$ by: $I(v_i) = 1$ if $s_i^t:=(0,0,0,0,i,v_i,v_i,a_k)$ is in $Z$, $I(v_i) = 0$ if $s_i^f:=(0,0,0,0,i,\neg v_i,\neg v_i,a_k)$ is in $Z$. Since $\mA \models_Z x \perp_z y$, $s_i^f(z) = s_i^t(z)$ and because there is no $s' \in X$ such that $s'(x) = s_i^f(x), s'(y) = s_i^t(y)$ and $s'(z) = s_i^f(z)$, for each $i$ at most one of $s_i^f$, $s_i^t$ can be in $Z$. Similarly, because $\mA \models_Y x \perp_t y$, only one of them can be in $Y$. Thus $I$ is indeed a function.

Since $\mA \models_Z x\perp_z y$ and there is no assignment in $X$ such that $(x,y)\mapsto(v_i,\neg v_i)$, every pair $s_1,s_2 \in Z$ such that $s_1(z) = s_2(z) =i$ must have the same value of $x$ and $y$. Every assignment representing a clause in $Z$ respects the choice of $I$. Furthermore, since $\mA \models_Z c_1 \perp_c c_2$ and $(w,c,c_1,c_2) \mapsto (1,i,0,0), (1,i,1,0),(1,i,0,1)$ are in $Z$, $(1,i,1,1)$ must be in $Z$, i.e., at least one assignment per clause is in $Z$. By the above we may conclude that $I$ satisfies $\Phi$: at least one literal per clause is evaluate to $1$ by $I$.
%\end{itemize}

Finally it is not difficult to check that  the splitting of $X$ into $Y$ and $Z$, if possible, can be  realized with disjoint $Y$ and $Z$. Hence the hardness holds  under both strict and lax semantics %. \juha{Last remark added.}

\end{proof}

It is worth noting that the formula in the previous result has disjunction-width two whereas the formula in Corollary~\ref{3 disjunction relativized independance} simulating dependence atoms by conditional independence atoms has width three. We will next show  an analog of these results for  pure (i.e., non-conditional) independence atoms. %The results holds both under strict and lax semantics.
%The following theorem further shows that even unary pure independence atoms suffice to make the model-cheking problem NP-complete.

\begin{theorem}\label{pure3-disjunction}
	The model checking problem 	is $\np$-complete for  
	$$\phi := \left(x \perp y\right) \vee \left(x \perp y\right) \vee \left(x \perp y\right) \vee x \neq y.$$
%	is $\np$-complete.
\end{theorem}

\begin{proof} Membership in $\np$ is obvious. Hardness is proved by reduction from the following $\np$-complete problem $3$-clique cover (see~\cite{GareyJ1979}): Given $G=(V_G,E_G)$ a finite graph, decide if there exists a collection of disjoint triangle subgraphs of $G$, that cover the vertex set $V_G$ of $G$.

	So, let $G=(V_G,E_G)$ be a finite graph, $\mA = V_G $ be a first-order structure of the empty signature, and $X$ the team
	 $$X = \{\,(v,v) \, : \, v \in V_G\,\}\cup \{\,(v_1,v_2), (v_2,v_1) \, : \, (v_1,v_2) \in E_G\,\},$$
	  where $(v_1,v_2)$  denotes the assignment $s$ with $s(x)=v_1$ and  $s(y)=v_2$.
		 We are going to show that $G$ has a $3$-clique cover if and only if $\mA \models_X \phi$.
	 
	 $$\begin{array}{|c|c|}
		\hline
		\col x & \col y \\
		\hline
		v_1 & v_1 \\
		v_2 & v_2 \\
		\vdots & \vdots \\
		v_{|V_G|} & v_{|V_G|} \\
		\hline
		v_1 & v_2 \\
		v_2 & v_1 \\
		v_2 & v_4 \\
		v_4 & v_2 \\
		\vdots & \vdots \\
		\hline
	\end{array}$$
	
%	 \begin{itemize}
%	 \item[$\Rightarrow$] 

Suppose that $G$ has a $3$-clique cover, i.e., there exists $C_1,C_2,C_3$ three cliques such that $V_G = V_{C_1} \cup V_{C_2} \cup V_{C_3}$. We have to prove  $\mA \models_X \phi$. For $i \in \{\,1,2,3\,\}$, let $$X_i = \{\,(v,v)\, : \, v \in C_i\,\} \cup  \{\,(v_1,v_2), (v_2,v_1)\, |\, v_1,v_2 \in C_i\,\}$$ and $X_4 = X\setminus (X_1\cup X_2 \cup X_3)$.
	 
	 Because it is a vertex cover, every assignment of the form $(v,v)$ is contained in $X_1\cup X_2\cup X_3$ and not in $X_4$, i.e. $\mA \models_{X_4} x \neq y$.
	 
	 Let $i \in \{\,1,2,3\,\}$ and $s,s' \in X_i$ be two assignments. If $s(x,y)=(v,v)$ and $s'(x,y)=(v',v')$, then $v,v' \in C_i$ and there exists two assignments $s_1,s_2$  in $X_i$ such that $s_1(x,y) = (v,v')$ and $s_2(x,y)=(v',v)$ by construction. Similarly if $s(x,y) = (v,v)$ and $s'(x,y) = (v_1,v_2)$, there exists in $X_i$ the assignments $(v,v_2)$ and $(v_1,v)$ (even if $v_1 = v$ or $v_2 = v$). Finally, if $s(x,y)=(v_1,v_2)$ and $s'(x,y)=(v_1',v_2')$, the assignments $(v_1,v_1),(v_2,v_2),(v_1',v_1'),(v_2',v_2')$ are in $X_i$ and so are $(v_1,v_2'),(v_1',v_2)$. The above implies that $\mA \models_{X_i} x \perp y$.
	 
%	\item[$\Leftarrow$]
For the converse implication,   suppose that $\mA \models_X  \left(x \perp y\right) \vee \left(x \perp y\right) \vee \left(x \perp y\right) \vee x \neq y$, then $X = X_1 \cup X_2 \cup X_3 \cup X_4$ such that $\mA \models_{X_i} x \perp y$ for $i \in \{\,1,2,3\,\}$ and $\mA \models_{X_4} x \neq y$.
	Let $C_i$,  for $i \in \{\,1,2,3\,\}$,  be the graph whose vertices are $\{\, v \, : \,(v,v) \in X_i\,\}$ and  edges are 
	 \[ \{\, (v_1,v_2)\, : \,(v_1,v_2) \in X_i \text{ and }(v_2,v_1) \in X_i\,\}. \]
	 Note that some $C_i$ can be empty but they form a vertex cover of $G$ as no assignment $(v,v)$ is in $X_4$. If $v,v'\in C_i$ then $(v,v)$ and $(v',v')$ are in $X_i$. By independence, $(v,v')$ and $(v',v)$ are also in $X_i$. Therefore the edge $(v,v')$ is in $C_i$: $C_i$ is a clique. Therefore, $G$ is covered by the three disjoint cliques $C_1' = C_1$, $C_2' = C_2 \setminus  C_1$ and $C_3' = C_3 \setminus (C_1 \cup C_2)$. 
	
		If we are in the strict semantics, the sets $X_i$ are disjoint and  so are the cliques. If we are in lax semantic, $G$ is covered by the disjoint cliques $C_1' = C_1$, $C_2' = C_2 \setminus C_1$ and $C_3' = C_3 \setminus (C_1 \cup C_2)$. In both cases, $G$ has a $3$-clique cover.
%	 \end{itemize}
		 
\end{proof}

%The formalisms above can be generalized in two ways: either by allowing existential quantification (with or without a limited form of disjunction), or by allowing more than one disjunction. We will prove that all choices lead to intractability of the model-checking problem.

\subsection{The case of quantified formulas}

In this section we show that  existential quantification even without disjunction and even in the empty vocabulary makes the model checking problem hard  for both dependence and independence logic. %Let $\AndExistDep$ and  $\AndExistInd$  be the fragments of $\df$ and $\Ind$ using existential quantification and conjunction only. 

\begin{theorem}\label{Exists+wedge} 
There exists a formula $\phi$ with $\np$-complete model-checking problem where $\phi\equiv \exists x \psi$ and $\psi$ is a conjunction of two dependence atoms.

%There is a formula $\phi$ of dependence logic of empty non-logical vocabulary built with $\exists $ and $\wedge$ whose model-checking problem is  $\np$-complete.
 
\end{theorem}

\begin{proof}
Before giving the proof, let us consider the following problem: Given  a graph $G$ with  $n^2$ vertices, are the vertices of $G$ colourable with $n$ colors (such that no adjacent vertices carries the same color). This problem is easily seen to be a generalization of the well-known $\np$-complete $3$-coloring problem (see~\cite{GareyJ1979}). Indeed, let $G=(V_0,E_0)$ be an undirected graph with $|V_0|=n$ vertices, let $K_{n-3}=(V_1,E_1)$ be the complete graph with $|V_1|=n-3$ vertices (hence $|E_1|=(n-3)(n-4)/2$ edges) and let $G'$ be the graph with vertex set $V'=V_0\cup V_1\cup V_2$ with $V_2$ of size $n^2-n-n+3$ (hence $|V'|=n^2$) and edge set $E_0\cup E_1 \cup \{\{x,y\}: x\in V_0, y\in V_1\}$. Then, it is easy to see that $G$ is $3$-colorable iff $G'$ is  
	$n$-colorable.

Let us now define  the formula $\phi$ as follows:
	$$\phi \equiv \exists x(\dep\left(x,r_1,r_2,e,m\right)\wedge \dep \left(v_1,v_2,x\right)).$$ 
	We will reduce the problem of determining whether a graph $G$ with  $n^2$ vertices is $n$-colorable to the model-checking problem of $\phi$. 

	Let $G$ be a graph with $n^2$ vertices $V_G=\{\alpha_0,\hdots,\alpha_{n^2-1}\}$, $\mA=\{0,\ldots, n-1\}$ a first order structure of the empty signature and $X = \{s_i^j \ |\  i \in \{0,\ldots ,n^2-1\},\  0 \le j\le i \}$ be a team such that :

	\begin{itemize}
		\item $s_i^j(v_1) = \lfloor i/n\rfloor$ and $s_i^j(v_2) = i \mod n$. %for every $j \in \{0,\ldots, i \}$. 
		In other words, $s_i^j(v_1,v_2)$ is the decomposition of $i$ in base $n$.
		\item $s_i^j(r_1) = \lfloor j/n\rfloor $ and $s_i^j(r_2) = j \mod n$. %, for every  $j \in \{0,\ldots, i \}$. 
		In other words, $s_i^j(r_1,r_2)$ is the decomposition of $j$  in base $n$.
		
		\item $s_i^j(m) = 0$ if $i \neq j$ and $s_i^i(m) = 1$.%, for any $j \in [1,i]$.
		\item $s_{j}^{i}(e) = 1$, if $j=i$, or if there is an edge between $\alpha_i$ and $\alpha_{j}$ with $j > i$. Otherwise $s_{j}^{i}(e) = 0$.
	\end{itemize}
	
	For example, for $n= 2$ and $E_G = \{\{0,1\},\{1,2\},\{0,2\},\{2,3\}\}$, we  obtain the following team on the universe $A= \{0,1\}$:
	
	$$\begin{array}{|c|c|c|c|c|c|c|}
	\hline
	x & v_1 & v_2 & r_1 & r_2 & m & e \\
	\hline
	& 0 & 0 & 0 & 0 & 1 & 1 \\
	\hline
	& 0 & 1 & 0 & 0 & 0 & 1 \\
	& 0 & 1 & 0 & 1 & 1 & 1 \\
	\hline
	& 1 & 0 & 0 & 0 & 0 & 1 \\
	& 1 & 0 & 0 & 1 & 0 & 1 \\
	& 1 & 0 & 1 & 0 & 1 & 1 \\
	\hline
	& 1 & 1 & 0 & 0 & 0 & 0 \\
	& 1 & 1 & 0 & 1 & 0 & 0 \\
	& 1 & 1 & 1 & 0 & 0 & 1 \\
	& 1 & 1 & 1 & 1 & 1 & 1 \\
	\hline
	\end{array}$$
	$$G \text{ is } n\text{-colourable iff } \mA\models_X \phi$$
	We are going to demonstrate that $\mA \models_X \phi$ if and only if $G$ is $n$-colourable.
	
	First the left to right implication. Since $\mA\models_X \phi$ there exists  a mapping  $F\colon X\rightarrow \mathcal{P}(A)\setminus \{\emptyset\}$ such that $\mA\models_{X(F/x)} \dep\left(x,r_1,r_2,e,m\right)\wedge \dep \left(v_1,v_2,x\right)$.  By downwards closure, we may assume without loss of generality that $F(s)$ is a singleton for all $s\in X$.
	% Let $X'$ be an extension of $X$ such that $\mA\models_{X'} \dep\left(x,r_1,r_2,e,m\right)\wedge \dep \left(v_1,v_2,x\right)$. 
	%Let $F: X \rightarrow A$ be defined by $F(s_i^j) = s_i^j(x)$.
Since  $\dep(v_1,v_2,x)$ holds,  $F$ induces a mapping  $F'\colon V_G \rightarrow A$, by 
$F'(\alpha_i) := a $, for the unique $a$ such that $F(s_i^0)=\{a\}$. If there is an edge between $\alpha_i$ and $\alpha_{i'}$, $i' > i$, then $s_{i'}^{i}(e) = 1 = s_i^{i}(e)$. Furthermore, $s_{i'}^{i}(r_1,r_2) = s_{i}^{i}(r_1,r_2)$ but $ s_{i'}^{i}(m) = 0$ and $s_{i}^{i}(m) = 1$. Therefore, because the atom  $\dep (x,r_1,r_2,e,m)$ holds,  we must have $F(s_{i}^{i}) \neq  F(s_{i'}^{i})$. Thus  $F'(\alpha_i) \neq F'(\alpha_{i'})$ if there is an edge between $\alpha_i$ and $\alpha_{i'}$.  This shows that $F'$ is a colouring of $G$ with $|A| =  n$ colours.
	
	Let us then consider the right to left implication. Let $c: V_G \rightarrow \{0,\ldots ,n-1\}$ be an $n$ colouring. We extend $X$ to variable $x$ with a new team $X'$ such that $s^j_i(x) = c(\alpha_i)$. The value of $x$ depends only on $i$, which is encoded in $(v_1,v_2)$, i.e., $\mA \models_{X'} \dep(v_1,v_2,x)$.
	
	 Let $s_i^j, s_{i'}^{j'}$ be two assignments of $X'$. Suppose that $s_i^j(r_1,r_2,e) = s_{i'}^{j'}(r_1,r_2,e)$ but $s_i^j(m) \neq s_{i'}^{j'}(m)$. In this case we must check that $s_i^j(x)$ is different from $s_{i'}^{j'}(x)$ (because $\mA \models_{X'} \dep(x,r_1,r_2,e,m)$). Now it holds that $j=j'$ because $s_i^j(r_1,r_2) = s_{i'}^{j'}(r_1,r_2)$. Furthermore, since $s_i^j(m) \neq s_{i'}^{j'}(m)$, either $i=j$ or $i'=j'$. Let us suppose $i=j$. Because $1= s^i_i(e)=  s^j_{i}(e) = s^{j'}_{i'}(e)=s^i_{i'}(e)$, there is an edge between $\alpha_i$ and $\alpha_{i'}$ in $G$. Therefore $c(i) \neq c(i')$ and $s^j_i(x) \neq s^{j'}_{i'}(x)$.
	\end{proof}

%\begin{corollary}
%The model-checking problem for  is $\np$-complete. 
%\end{corollary}
By encoding dependence atoms in terms of conditional independence atoms we get the analogous results for free for independence logic.
\begin{corollary}
There is a formula $\phi$ of independence logic of empty non-logical vocabulary built with $\exists $ and $\wedge$ whose model-checking problem is  $\np$-complete.
 
\end{corollary}
We will next show that this corollary can be strengthened in the case of independence logic under the strict semantics. In other words, we will now show a  version of Theorem \ref{Exists+wedge}  for pure independence atoms under the strict semantics. Let $\psi(t,c,v)$ be the following formula over signature $\sg=\{R\}$, where $R$ is a ternary relation symbol and  $t,c$ and $v$ are free variables:

\[  \psi(t,c,v) := \exists x (t \perp x \wedge R(c,v,x)).\]

\begin{proposition}\label{StrictInd}
For all propositional formulas $\phi$ in 3-\pb{cnf}, one can compute in polynomial time a  team $X$, with $\Dom(X)=\{t,c,v\}$ and a structure $\mA$ such that:

\[ \phi \mbox{ is satisfiable } \iff  \mA \models_X \psi(t,c,v),\]
under the strict semantics.
\end{proposition}

\begin{proof}
Without loss of generality, let $\phi=\bigwedge_{i=1}^m C_i$ be a $3$-\pb{cnf} formula over a set $V=\{v_1,\ldots,v_n\}$ of variables of size $n$. Let $C_i=l_{i_1}\vee l_{i_2}\vee l_{i_3}$, with $l_{i_j}\in\{v_1,...,v_n,\neg v_1, ..., \neg v_n\}$. We first describe the  relation $R$ built on the domain $D=\{0,\ldots, m+n,v_1,...,v_n,\neg v_1, ..., \neg v_n\}$:

\[
\begin{array}{l}
R = 
\ \{(0,i,v_i), (0,i,\neg v_i): 1\leq i \leq n\} \cup \\ 
 \ \{(i,0,l_{i_1}), (i,0,l_{i_2}),  (i,0,l_{i_3}): 1\leq i \leq m, \ C_i=l_{i_1}\vee l_{i_2}\vee l_{i_3}\}  \cup\\
 \  \{(m+i,0,v_i), (m+i,0,\neg v_i): 1\leq i \leq n\}. 
\end{array} 
\]
Let $\mA=(D,R)$.
Finally, team $X$ is the union of the two assignment sets $Y$ and $Z$ below:

 \[
 \begin{array}{cc}
  Y:
  \begin{array}{|c|c|c|}
    \hline  c & v  & t \\ \hline
  \hline  0 & 1  & 0 \\ \hline
  0 & 2  & 0 \\ \hline
 \vdots & \vdots  & \vdots \\ \hline
   0 & n & 0 \\ \hline
  \end{array} 
& 
  Z:
  \begin{array}{|c|c|c|}
    \hline  c & v  & t \\ \hline
    1 & 0  & 1 \\ \hline
     2 & 0  & 1 \\ \hline
    \vdots & \vdots  & \vdots \\ \hline
    m & 0 & 1 \\ \hline
    m+1 & 0 & 1 \\ \hline
      \vdots & \vdots  & \vdots \\ \hline
      m+n & 0 & 1 \\ \hline   
  \end{array} 
  \end{array}
  \]

Variable $t$ encodes the type of the object in consideration: $0$ for a clause, $1$ for a Boolean variable. The first $n$ assignments deal with variables (hence the value of $c$ is set to $0$, by convention), the last $m$ assignments deal with clauses (hence, $v$ is set to $0$). It now remains to show that 

\[ \phi \mbox{ is satisfiable } \iff  \mA \models_X \psi(t,c,v),\]
where  $\psi(t,c,v)$ is the formula $\exists x \ (t \perp x \wedge R(c,v,x))$.
Suppose first that $\phi$ is satisfiable and let $I:V\rightarrow \{0,1\}$ be such that $I(\phi)=1$. Let $F:X\rightarrow D$ be such that: 
%\arnaudmargin{The argumentation needs to be cleaned. I am rather convinced also a binary relation can be used instead of a ternary one}
\begin{enumerate}
\item\label{item1-exists-indep} if $s(c)=0$ and $s(v)=i$,  then $F(s)=v_i$ if $I(v_i)=1$ and $F(s)=\neg v_i$ if $I(v_i)=0$. Similarly, if $s(c)=m+i$, $i\geq 1$, then  $F(s)=v_i$ if $I(v_i)=1$ and $F(s)=\neg v_i$ if $I(v_i)=0$.

\item\label{item2-exists-indep} if $s(c)=i$, $1\leq i\leq m$, then $F(s)=l_{i_j}$, for $j\leq 3$, such that $I(l_{i_j})=1$. Such a $j$ always exists since $I\models \phi$.
\end{enumerate}

Let $X'=X(F/x)$. It is clear that for all $s\in X'$, $\mA \models_s R(c,v,x)$ holds by the construction. In $X'$, variable $t$ takes only two values $0$ and $1$. Let now

\[ 
\begin{array}{l}
V_0=\{s(x): s\in X' \mbox{ and } s(t)=0  \}, \\
 V_1=\{s(x): s\in X' \mbox{ and } s(t)=1  \}.
\end{array}
\]

Now to  show the claim  $\mA \models_X \psi(t,c,v)$, it remains to  show $\mA\models _{X'} t \perp x$. For this it suffices to show that $V_0=V_1$. Note that, for all $i\leq n$, either $v_i$ or $\neg v_i$ belongs to $V_0$. Also, by the construction, $V_0\subseteq V_1$. Indeed, by item~(\ref{item1-exists-indep}), for any $s$ with $s(t)=s(c)=0$ and  $s(v)=i$ it holds that $s(x)=s'(x)$, where $s'(c)=m+i$ and $s'(t)=1$.
Suppose now that there exists $s\in X'$ such that $s(x)\in V_1$ and $s(x)\not\in V_0$. Clearly, for such an $s$,  $s(c)=i$, for some $1\leq i\leq m$. Then, by the construction of the function $F$, $s(x)=l_{i_j}$, for $j\leq 3$, such that $I(l_{i_j})=1$. But then again by the definition of $F$, $s(x)\in V_0$ which is a contradiction. Therefore, $V_0=V_1$, $\mA\models _{X'} t \perp x$, and hence  $\mA \models_X \psi(t,c,v)$.

We now prove the other implication.
%\juha{I looks to me that lax semantics of the existential quantifier might cause us problems with this direction: Let $f(s)$ be the set  $\{v_i,\neg v_i\}$ for those assignments $s$ such that $s(c)=0, s(v)=i$ or $s'(c)=m+i, s'(v)=0$. Then the independence atom $t \perp x$ will be satisfied regardless of what values $f$ assigns to assignments with $s(c)=i\in \{1,....,m\}$.}
Suppose  $\mA\models_X\psi(t,c,v)$ and let $X'=X(F/x)$ be such that

 \[ \mA\models_{X'}  t \perp x \wedge R(c,v,x). \]
 
 Because $\mA\models_{X'} t \perp x$ and $s(t)\in \{0,1\}$ for all $s\in X'$, it holds that $V_0=V_1$, where $V_0$ and $V_1$ are as defined in the first part of the proof above. Together with the definition of the relation $R$, this implies that if $s,s'\in X$ such that $s(c)=0, s(v)=i$ and $s'(c)=m+i, s'(v)=0$, then $s(x)=s'(x)\in \{v_i,\neg v_i\}$. 
 Define now the Boolean assignment $I:V\rightarrow \{0,1\}$ by
 
 \[ I(v_i)=1 \iff s(x)=v_i \mbox{ for } s \mbox{ s.t. } s(c)=0, s(v)=i.\]

Now consider $s\in X'$ with $s(c)=i\in \{1,...,m\}$. By definition of $R$, $s(x)=l_{i_j}$ (for some $j\leq 3$). Since $V_1\subseteq V_0$ there must exists $s'\in X'$ s.t. $s'(t)=0$, $s'(x)=l_{i_j}$ and, also immediately by the construction of $X$, $s'(c)=0$ and $s'(v)=i$. The definition of $I$ implies that $I(l_{i_j})=1$ and that clause $C_i$ is satisfied.
\end{proof}

We end this section by noting that existential quantifiers cannot be replaced by universal quantifiers in the above theorems.

\begin{proposition}\label{forall+wedge} The model-checking problem for formulas of  dependence or independence  logic using only  universal quantification and conjunction is in $\logspace$.
% Let $\phi$  be a formula of  dependence or independence  logic using only  universal quantification and conjunction. Then the  model-checking problem of $\phi$  is in L.
\end{proposition}
%\begin{proof}[Sketch] It is easy to show that $\phi$ can be expressed in FO assuming  the team $X$ with domain $\tu x$ is represented by the $n$-ary relation $X(\tu x)$. 
%\end{proof}

\begin{proof} %By the assumption, $\phi$ may only contain universal quantifiers and conjunction. 
Given $\phi$, we first transform it into prenex normal-form exactly as in first-order logic \cite{vaananen07}. We may hence assume that $\phi $ has the form
\[  \forall x_1\ldots \forall x_n \bigwedge \theta_i (x_1,\ldots, x_n,y_1,\ldots,y_m),        \]
where $\theta _i$ is either a first-order, dependence, or independence atom. Let  $\mA$ be a model,  and  $X$ be a team of $A$ with domain $\{x_1,\ldots,x_n,y_1,\ldots,y_m\} $. As in   \cite{vaananen07}, the formula $\bigwedge \theta_i (x_1,\ldots,x_n,y_1,\ldots,y_m)$ can be expressed by a first-order sentence $\psi$ when the team $X$ is represented by the $n+m$-ary relation $X(\tu x,\tu y)$, that is,
$$ \mA\models _X \bigwedge \theta_i (x_1,\ldots,x_n,y_1,\ldots,y_m) \Leftrightarrow (\mA,X(\tu x,\tu y))\models \psi. $$
Since $X(\tu x,\tu y)$ is a first-order definable extension of  $X(\tu y)$ it is clear that we can construct a FO-sentence $\psi'$ such that 
$$ \mA\models _X  \forall \tu x \bigwedge \theta_i (\tu x,\tu y) \Leftrightarrow (\mA,X(\tu y))\models \psi', $$
holds for all structures $\mA$ and teams $X$ with domain $\{y_1,\ldots,y_m\}$. The claim follows from the fact that the data complexity of FO is in $\mathrm{L}$.
\end{proof}

\section{Inclusion logic under the lax semantics}

Recall that a \pb{cnf} formula $\Phi$ is called dual-Horn if each of its clauses contains at most one negative literal. The satisfiability problem of dual-Horn \pb{cnf} formulas, \pb{dual-horn-sat}, is known to be $\Ptime$-complete (see~\cite{GareyJ1979}).  

 In this section we show that the model-checking problem of inclusion logic formulas under the lax semantics  can be reduced to \pb{dual-horn-sat}.
  
For a team $X$, $\tu x= \langle x_{i_1},...,x_{i_n}\rangle \in \dom(X)^n$, and $s\in X$, we denote by $\restrict{s}{\tu x}$ the restriction of $s$ to the  variables $x_{i_1},...,x_{i_n}$. In this section, $\sg$ denotes a relational signature.

  \begin{proposition}\label{FO inclusion lax}
  There exists an algorithm which, given 
   $\phi\in\FOinclusion$,  a structure $\mA$ over $\sigma$, and  a team $X$ such that $\Var({\phi})\subseteq \dom(X)$, outputs a propositional formula $\Psi$ in dual-Horn form such that: $\mA \models_X \phi \iff  \Psi$ is satisfiable.
   %   \[ \mA \models_X \phi \mbox{ iff } \Psi \mbox{ is satisfiable. } \] 
   Furthermore, when $\phi$ is fixed, the algorithm runs in logarithmic space in the size of $\mA$ and $X$.
  \end{proposition}
  
  \begin{proof}
%{read r from X throughout the proof}  
Let $\phi, \mA, X$ be as above. For any team $X$, we  will consider the set $\mX$ of propositional variables $X[s]$ for $s\in A^{\dom(X)}$.
  % and the valuation $I:\mX\longrightarrow \{0,1\}$ defined by:
%\[ I(X[s])=1 \iff s\in X. \]
Starting from $\phi$, $\mA$, and $X$ we decompose step by step the formula $\phi$ into subformulas (until reaching its atomic subformulas) and different teams $Y$, $Z$, ... and control the relationships between the different teams by propositional dual-Horn formulas built over the propositional variables issued from $X,Y,Z,...$..
Let $\calS=\{(\phi,X,V)\}$, where $V=\dom(X)$, and $\calC=\{X[s] : s\in X\} \cup \{\neg X[s] : s\not\in X\}$. The propositional formula $\Psi$ is now constructed inductively as follows.
  
  As long as $\calS\neq \emptyset$, we apply the following rule:  Pick $(\phi,X,V)$ in $\calS$ and apply the following rules.
  
  \begin{itemize}
  	\item If $\phi$ is  $R(\tu x)$ with $R$ a literal of $\sg$ then: $\calS:=\calS\backslash \{(\phi,X,V)\}$  and 
  	$\calC:=\calC \cup \{X[s]\rightarrow \top : s \in A^V\mbox{ and } \mA\models_X R(\tu x) \}\cup \{X[s]\rightarrow \bot : s \in A^V\mbox{ and } \mA\not\models_X R(\tu x) \}$.
  	%  \[\calS:=\calS\backslash \{(\phi,X,r)\} \mbox{ and } \calC:=\calC \cup \{X[s]\rightarrow R(\restrict{s}{\tu x}): \mbox{ for all } s\in A^r \}
%  \]
%  	
%  		\begin{eqnarray*}
%  	&  	\calS:=\calS\backslash \{(\phi,X,r)\} \mbox{ and }\\
%  	&  	\calC:=\calC \cup \{X[s]\rightarrow R(\restrict{s}{\tu x}): \mbox{ for all } s\in A^r \}
%  	  	\end{eqnarray*}
%  	
%  	\[ \calS:=\calS\backslash \{(\phi,X,r)\} \mbox{ and }\] 
%  	
%  	\[\calC:=\calC \cup \{X[s]\rightarrow R(\restrict{s}{\tu x}): \mbox{ for all } s\in A^r \}  \]
% \juha{$X[s]\rightarrow R(\restrict{s}{\tu x})$ should be reformulated and also the definition  of $\calC$ above.} 	
% \arnaud{I have reformulated the definition of $\calC$ but maybe $X[s]\rightarrow R(\restrict{s}{\tu x})$ is clearer than $\bigcup_{s(\tu x)\not\in R}\neg X[s]$}
Clearly, it holds that $\mA\models_XR(\tu x)$ iff $ \bigwedge\calC$ is satisfiable.

  	\item If $\phi$ is  $\tu x \subseteq \tu y$ then: $\calS:= \calS\backslash \{(\phi,X,V)\}$  and  
  	\[\textstyle
  		%\calS:= \calS\backslash \{(\phi,X,r)\}  \mbox{ and } 
  		\calC:=\calC \cup \{X[s]\rightarrow \bigvee_{s'\in A^V, \restrict{s'}{\tu y}=\restrict{s}{\tu x}} X[s']:
  	 s\in A^V \}. \]
  	 %  	  	 		\begin{eqnarray*}
%  	  	  	&  	\calS:= \calS\backslash \{(\phi,X,r)\}  \mbox{ and } \\
%  	  	  	&  	 \calC:=\calC \cup \{X[s]\rightarrow \bigvee_{s'\in A^r, \restrict{s'}{\tu y}=\restrict{s}{\tu x}} X[s']: %\mbox{ for all }
%  	  	  	 s\in A^r \} 
%  	  	  	  	\end{eqnarray*}
%  	
%  	\[ \calS:= \calS\backslash \{(\phi,X,r)\}  \mbox{ and } \calC:=\calC \cup \{X[s]\rightarrow \bigvee_{s'\in A^r, \restrict{s'}{\tu y}=\restrict{s}{\tu x}} X[s']: \mbox{ for all } s\in A^r \}  \]
%  	
  	It holds that $\mA\models_X\tu x \subseteq \tu y$ iff $ \bigwedge_{s\in A^V} (X[s]\rightarrow \bigvee_{s'\in A^V, \restrict{s'}{\tu y}=\restrict{s}{\tu x}} X[s'])$ is satisfiable.

  	\item If $\phi$ is  $\exists x \psi$, then: $\calS := ( \calS \backslash \{(\phi,X,V)\} ) \cup \{(\psi,Y, V\cup\{x\})\}$ and 
  	\[
  		\textstyle
%  	 	\calS := ( \calS \backslash \{(\phi,X,r)\} ) \cup \{(\psi,Y, r+1)\}  \mbox{ and } 
  	    	\calC:=\calC \cup \{X[s]\rightarrow \bigvee_{
  	  	  	a\in A} Y[s(a/x)]: s\in A^V \}, 
  	  	\]
  	%  	  		\begin{eqnarray*}
%  	  	  	&  	\calS := ( \calS \backslash \{(\phi,X,r)\} ) \cup \{(\psi,Y, r+1)\}  \mbox{ and } \\
%  	  	  	&  	\calC:=\calC \cup \{X[s]\rightarrow \bigvee_{
%  	  	  	  	%s'\in A^{r+1}, \restrict{s'}{\tu x}=\restrict{s}{\tu x}
%  	  	  	  	s'=(s,a),\  a\in A} Y[s']: s\in A^r \} 
%  	  	  	  	\end{eqnarray*}
  	  	%  	\[ \calS := ( \calS \backslash \{(\phi,X,r)\} ) \cup \{(\psi,Y, r+1)\} \]
%  	
%  	
%  	\[ \calC:=\calC \cup \{X[s]\rightarrow \bigvee_{
%  	%s'\in A^{r+1}, \restrict{s'}{\tu x}=\restrict{s}{\tu x}
%  	s'=(s,a),\  a\in A} Y[s']: \mbox{ for all } s\in A^r \}   \]
%  	
\noindent where the $Y[s]$, $s\in A^{V\cup\{x\}}$ are new  propositional variables (not used in $\calC$).
If $\mA\models_X\exists x \psi$ then, there exists a function $F\colon X\rightarrow \mathcal{P}(A)\setminus \{\emptyset\}$, such that $\mA\models_{X(F/x)} \psi$. In other words,  $\mA\models_Y\psi$ for some team $Y$ defined by the solutions  of the constraint $\bigwedge_{s\in A^V} X[s]\rightarrow \bigvee_{a\in A} Y[s(a/x)]$ (which define a suitable function $F$). Conversely, if $\mA\models_Y\psi$ for a team $Y$ as above defined from $X$, then clearly $\mA\models_X\exists x \psi$.

  	\item If $\phi$ is   $\forall x \psi$, then: $\calS := ( \calS \backslash \{(\phi,X,V)\} ) \cup \{(\psi,Y, V\cup\{x\})\} $ and 
  	  		\[
 \begin{array}{l}
%  \calS := ( \calS \backslash \{(\phi,X,r)\} ) \cup \{(\psi,Y, r+1)\}   \mbox{ and } \\
  \calC :=  \calC \cup  \{X[s]\rightarrow Y[s']:  s\in A^V, s'\in A^{V\cup\{x\}} \mbox{ s.t. } \restrict{s'}{\tu x}=\restrict{s}{\tu x} \},  
 \end{array}
  	  	\]
  	  	%  	  	\[ \calS := ( \calS \backslash \{(\phi,X,r)\} ) \cup \{(\psi,Y, r+1)\} \]
%  	  	
%  	
%  	 
%  	  	\[ \calC:=\calC \cup \{X[s]\rightarrow Y[s']: \mbox{ for all } s\in A^r, s'\in A^{r+1} \mbox{ s.t. } \restrict{s'}{\tu x}=\restrict{s}{\tu x} \}   \]
  	  	
  	  	 \noindent where the $Y[s]$, $s\in A^{V\cup\{x\}}$ are new  propositional variables (not used in $\calC$). The conclusion is similar as for the preceding case.
  	  	 
  	\item  If $\phi$ is   $\psi_1 \wedge \psi_2$ then: $\calS := ( \calS \backslash \{(\phi,X,V)\} ) \cup \{(\psi_1,X, V), (\psi_2,X, V)\}$
  	and $\calC$ is unchanged. By definition, $\mA\models_X\phi$ iff $\mA\models_X\psi_1 \wedge \psi_2$.
  	 \item 	If $\phi$ is   $\psi_1 \vee \psi_2$ then: $\calS := ( \calS \backslash \{(\phi,X,V)\} ) \cup \{(\psi_1,Y, V), (\psi_2,Z, V)\} $ and
  	 %let $\mY=\{Y[s]: s\in A^{r}\}$, $\mZ=\{Z[s]: s\in A^{r}\}$ and:
  	 	$$\textstyle
  	 % \begin{array}{l}
  	   %\calS := ( \calS \backslash \{(\phi,X,r)\} ) \cup \{(\psi_1,Y, r), (\psi_2,Z, r)\}   \mbox{ and } \\
  	   \calC :=  \calC \cup 
 \{X[s]\rightarrow Y[s]\vee Z[s]: s\in A^V \}\cup 
 \{Y[s]\rightarrow X[s],  Z[s]\rightarrow X[s]: s\in A^V \} 
  	  %\end{array}
  	   	  	$$
  	   	  	%  	 \[ \calS := ( \calS \backslash \{(\phi,X,r)\} ) \cup \{(\psi_1,Y, r), (\psi_2,Z, r)\}\]
%  
%  and 
%  
%\[   \calC:=\calC \cup \{X[s]\rightarrow Y[s]\vee Z[s], Y[s]\rightarrow X[s],  Z[s]\rightarrow X[s]: \mbox{ for all } s\in A^r \} 
% \]  
   \noindent where again the $Y[s]$ and $Z[s]$, $s\in A^{V}$ are  new propositional variables  (not used in $\calC$). Here again, $\mA\models_X\phi$ if and only if $\mA\models_Y\psi_1$ and $\mA\models_Z\psi_2$ for some suitable $Y$ and $Z$ such that $Y\cup Z=X$ which is exactly what is stated by the Boolean constraints.
  \end{itemize}
  
  Observe that each new clause added to $\calC$ during the process is of dual-Horn form, i.e., contains at most one negative literal. Observe also, that applied to some $(\phi,X,r)$, the algorithm above only adds triples  $\calS$ whose first component is a proper subformula of $\phi$ and eliminates  $(\phi,X,r)$. When the formula $\phi$ is atomic, no new triple is added afterwards.  Hence the algorithm will eventually terminate with $\calS=\emptyset$. 
 Setting $\Psi:= \bigwedge_{C\in\calC} C$, it can easily be proved by induction that: $\mA \models_X \phi$  iff  $\Psi$ is satisfiable.

   %\[ \mA \models_X \phi \mbox{ iff } \Psi \mbox{ is satisfiable. } \]  

Observe also that each clause in $\calC$ can be constructed from $X$ and $\mA$ by simply running through their elements (using their index) hence  in logarithmic space. 
  \end{proof}
  
 \begin{remark} 
 %Remark that t
 The construction of Proposition~\ref{FO inclusion lax} can  be done in principle for any kind of atom: dependence, independence, exclusion, constancy etc. 
 %But the resulting formula seems to be obviously in Dual-Horn form only in the case of inclusion. 
  To illustrate this remark, one could translate in the above proof a dependence atom of the form $\dep(\tu x,y)$ by (using the notations of the proof):
 
\[\bigwedge_{ \substack{s,s'\in A^r\\ s(\tu x)=s'(\tu x) \wedge s(y)\neq s'(y)}} (\neg X[s]\vee \neg X[s']).\]
 
 The additional clauses are of length two. A similar treatment can be done for independence atoms $\indep{\tu x}{\tu y}{\tu z}$.  In the two cases however, the resulting formula is not in Dual-Horn form anymore and there is no way to do so (unless $\Ptime=\np$).
 \end{remark}

  Since deciding the satisfiability of a propositional formula in dual-Horn form can be done in polynomial time we obtain the following already known (\cite{gallhella13}) corollary.
  
\begin{corollary}
  The data complexity of $\FOinclusion$ under the lax  semantics  is in $\Ptime$.
  \end{corollary}
  
  %The $\FOinclusion$ logic  is in fact known to be $\Ptime$-complete since it can express alternating graph reachability.

\section{Inclusion logic under the strict semantics}
In this  section we consider model-checking of inclusion logic formulas  under the strict semantics. % The first results is a strong version of Theorem \ref{Exists+wedge}  for pure independence atoms,  and the second result is the analogue of Theorem \ref{Exists+wedge} for inclusion logic.
 By the result of \cite{galhankon13}, inclusion logic with the strict semantics is equi-expressive with dependence logic. The following theorem shows that NP-completeness can be attained with quite simple formulas as in Theorem \ref{Exists+wedge}  combining strict existential quantification and inclusions atoms.
\begin{proposition}\label{StrictInc1}
	Let $\psi(c,v)$ be the following formula over signature $\sigma = \{R\}$:
	$$ \psi(c,v) = \exists x \exists y (y \subseteq x \wedge R(c,v,x,y)).$$
	For all propositional formulas $\phi$ in $3$-\pb{cnf}, one can compute in polynomial time a team $X$ with domain $\{c,v\}$ and a structure $\mA$ such that: $\phi$  is satisfiable  $\Leftrightarrow \mA \models_X \psi(c,v)$ under the strict semantics.
	%$$\phi \text{ is satisfiable } \Leftrightarrow \mA \models_X \psi(c,v)$$
\end{proposition}

\begin{proof}
	Let $\phi = \bigwedge_{i=1}^mC_i$ be a $3$-\pb{cnf} formula over a set $V = \{v_1,\hdots,v_n\}$ of variables. Let $C_i = \ell_{i_1} \vee \ell_{i_2} \vee \ell_{i_3}$, with $\ell_{i_j} \in \{v_1,\hdots,v_n,\neg v_1,\hdots, \neg v_n\}$.
	
	The domain of the structure $\mA$ is $A = \{0,\hdots,n,v_1,\hdots,v_n,\neg v_1,\hdots, \neg v_n\}$. The relation $R$ in this structure is:
	
	\begin{align*}
	R = & \{(0,i,v_i,0), (0,i,\neg v_i,0) | 1 \leq 1 \leq n\} \cup \\
		&\{(i,0,0, \ell_{i_1}),(i,0,0, \ell_{i_2}),(i,0,0, \ell_{i_3}) | 1 \leq i \leq m\}
	\end{align*}
	
	Finally, the team $X$ isgiven by the following table: 
		\[X=
		\begin{array}{|c|c|}
    			\hline  
		c & v  \\ 
			\hline
  		0 & 1   \\ 
			\hline
 		\vdots & \vdots \\ 
			\hline
		0 & n \\
			\hline
		1 & 0 \\  
			\hline
 		\vdots & \vdots  \\ 
			\hline
		m & 0 \\
			\hline
  		\end{array} 
		\]
	
	Now we claim that $\phi$ is satisfiable, if and only if $ \mA \models_X \exists x \exists y (y \subseteq x \wedge R(c,v,x,y))$.
	
%	\begin{itemize}
%		\item[$\Leftarrow$] 
Let us suppose that $ \mA \models_X \exists x \exists y (y \subseteq x \wedge R(c,v,x,y))$. Then there is an extension $X'$ of $X$ to the variables $x$ and $y$ such that 
		$ \mA \models_{X'} x \subseteq y \wedge R(c,v,x,y)$. Note that $X'$ 
	 has, e.g.,  the following shape:
		\[X':
		\begin{array}{|c|c|c|c|}
    			\hline  
		c & v & x & y \\ 
			\hline
  		0 & 1 & \neg v_1 & 0  \\ 
			\hline
 		\vdots & \vdots &\vdots & \vdots \\ 
			\hline
  		0 & i & \neg v_i & 0   \\ 
			\hline
 		\vdots & \vdots &\vdots & \vdots \\ 
			\hline
		0 & n  & v_n & 0\\
			\hline
		1 & 0 & 0 & \ell_{1_2} \\  
			\hline
 		\vdots & \vdots &\vdots & \vdots  \\ 
			\hline
		i & 0 & 0 & \ell_{i_1} \\  
			\hline
 		\vdots & \vdots &\vdots & \vdots  \\ 
			\hline
		m & 0 & 0 & \ell_{m_2} \\
			\hline
  		\end{array} 
		\]

Note that for each row of $X$ the variables $x$ and $y$ get exactly one value in $X'$.

Let $I: V \rightarrow \{0,1\}$ be the following assignment  of the variables of $\phi$: let $s \in X$  such that $s(v) = i$; we set $I(v_i) = 1$ if $s(x) = v_i$, $I(v_i) = 0$ otherwise (i.e., if $s(x) = \neg v_i$). Because $\mA \models_{X'} R(c,v,x,y)$ such $s$ exists and $s(x) \in \{v_i, \neg v_i\}$. Furthermore, $s$ is the only assignment in $X'$ such that $s(v) = i$ so there is no ambiguity in the definition of $I$.
		
Now we have to check that for every  clause of $\phi$ there is a literal which is evaluated to $1$ by $I$. Let $C_i$ be a clause of $\phi$ and $s$ the element of $X'$ such that $s(c) = i$.  Since $s(y)$ is a literal of $C_i$ and  $\mA \models_{X'} y \subseteq x$, there exists $s' \in X'$ such that $s'(x) = s(y)$. But $I\left(s\left(x\right)\right) = 1$ by definition, thus a literal of $C_i$ is evaluated to $1$ by $I$ and hence $I$ satisfies $\phi$.
	
%\item[$\Rightarrow$] 
Suppose then that $\phi$ is satisfiable and let $I: V \rightarrow \{0,1\}$ an assignment of the variables of $\phi$ which satisfies $\phi$. We extend $X$ to a team $X'$ over variables $\{c,v,x,y\}$ as follows: for $1 \leq i \leq n$, if $s \in X$ is such that $s(v) = i$ we set $s(x) = v_i$ if $I(v_i) = 1$, and  $s(x) = \neg v_i$ otherwise. Furthermore, $s(y) = 0$.
		
		For $1 \leq i \leq m$, if $s \in X$ is such that $s(c) = i$, we set $s(x) = 0$ and $s(y) = \ell_{i_j}$ where $\ell_{i_j}$ is a literal of $C_i$ which is evaluated to $1$ by $I$. It is now easy to check that 
		$\mA \models_{X'} y\subseteq x \wedge R(c,v,x,y)$, and hence  $ \mA \models_X \exists x \exists y (y \subseteq x) \wedge R(c,v,x,y)$. 
		%$X'(x)$ is the set of literals built on $V$ and evaluate to $1$ by $I$. $X'(y)$ is a set of literals of clauses of $\phi$ evaluates to $1$ by $I$. Therefore  $X'(y) \subseteq X'(x)$. 
%	\end{itemize}
\end{proof}

The next proposition shows that $\np$-completeness can be also attained by combining strict disjunction with inclusions atoms.
\begin{proposition}\label{StrictInc2}
	There exists formulas $\phi_1,\phi_2,\phi_3$ built with $\subseteq, \wedge$ such that the model checking problem for $\phi_1 \vee \phi_2 \vee \phi_3$ under the strict semantics is NP-complete.
	
\end{proposition}

\begin{proof} Membership in $\np$ is obvious. For hardness, we exhibit a polynomial time reduction to $\pb{1-in-3-sat}$. 

Let $\psi$ be the following formula over variables $\{l,v,c,a,b,\one,\two,\three\}$:
\[	
	\begin{array}{rl}
	\psi  \equiv & (\ell \subseteq v \wedge b \subseteq c \wedge a \subseteq \one) \vee \\
	& ( b \subseteq c \wedge a \subseteq \two) \vee \\
	& ( b \subseteq c \wedge a \subseteq \three)) \\
	\end{array}
\]
%	$$\psi \equiv(\ell \subseteq v \wedge b \subseteq c \wedge a \subseteq \one) \vee(\ell \subseteq v \wedge ((b \subseteq c \wedge a \subseteq \two) \vee ( b \subseteq c \wedge a \subseteq \three)) $$
    
	We will show that for all   positive $3$-\pb{cnf} formulas $\phi$, one can compute in polynomial time a team $X$ and a structure $\mA$ such that:
	$$\phi \text{ is an instance of \pb{1-in-3-sat} } \Leftrightarrow \mA \models_X \psi(l,v,c,a,b,\one,\two,\three)$$
Let $\phi = \bigwedge_{i=1}^m C_i$ be a positive $3$-\pb{cnf} formula over a set $V = \{v_1, \hdots, v_n\}$ of variables. Let $C_i = l_{i_1} \vee l_{i_2} \vee l_{i_3}$.  Recall that $\phi$ is an instance of  the problem $1$-in-$3$-SAT if and only if there is a truth assignment such that each  clause of $\phi$ has exactly one true variable.
	
The domain of the structure $\mA$ is $D = \{v_1, \hdots, v_n, 0,\hdots, m\}$. The team $X$ is $Y \cup Z$, see Table~\ref{Y}.
\begin{table}
	\[\begin{array}{c}
	Y:
	\begin{array}{|c|c|c|c|c|c|c|c|c|}
    			\hline  
		\ell & v & c & a & b & \one & \two & \three \\ 
			\hline
			\hline
  		0 & v_1  & 0 & 0 & 0 & 0 & 0 & 0\\ 
			\hline
 		\vdots & \vdots & \vdots & \vdots & \vdots & \vdots & \vdots & \vdots \\ 
			\hline
		0 & v_n  &  0 & 0 & 0 & 0 & 0 & 0\\
			\hline
		\ell_{1_1} &  0 & 1 & 0 & 0 & 0 & 0 & 0\\
			\hline
  		\ell_{1_2} &  0 & 1 & 0 & 0 & 0 & 0 & 0\\
			\hline
  		\ell_{1_3} &  0 & 1 & 0 & 0 & 0 & 0 & 0\\
			\hline
		\vdots &  \vdots & \vdots & \vdots & \vdots & \vdots & \vdots & \vdots \\ 
			\hline
  		\ell_{i_1}  & 0 & i & 0 & 0 & 0 & 0 & 0\\
			\hline
  		\ell_{i_2}  & 0 & i & 0 & 0 & 0 & 0 & 0\\
			\hline
  		\ell_{i_3}  & 0& i & 0 & 0 & 0 & 0 & 0\\
			\hline
		\vdots  & \vdots & \vdots & \vdots & \vdots & \vdots & \vdots & \vdots \\ 
			\hline
  		\ell_{m_1}  & 0 & m & 0 & 0 & 0 & 0 & 0\\
			\hline
  		\ell_{m_2}  & 0 & m & 0 & 0 & 0 & 0 & 0\\
			\hline
  		\ell_{m_3}  & 0 & m & 0 & 0 & 0 & 0 & 0\\
			\hline
  	\end{array} 
  	\\
	\\
	Z:
	\begin{array}{|c|c|c|c|c|c|c|c|c|}
			\hline
		\ell & v & c & a & b & \one &\two&\three\\ 
			\hline
			\hline
    		0 & 0 &  0 & 1 & 1 & 1 & 2 & 3\\
			\hline
		\vdots &  \vdots & \vdots & \vdots & \vdots& \vdots & \vdots & \vdots \\ 
			\hline
		0 & 0 &  0 & 1 & m & 1 & 2 & 3\\
			\hline
		0 & 0  & 0 & 2 & 1 & 1 & 2 & 3\\
			\hline
		 \vdots & \vdots & \vdots & \vdots & \vdots &  \vdots & \vdots & \vdots \\ 
			\hline
		0 &  0 & 0 & 2 & m & 1 & 2 & 3\\
			\hline
		0 &  0 & 0 & 3 & 1 & 1 & 2 & 3\\
			\hline
		\vdots  & \vdots & \vdots & \vdots & \vdots &  \vdots & \vdots & \vdots\\ 
			\hline
		0 & 0 & 0 & 3 & m & 1 & 2 & 3\\
			\hline
  	\end{array} 

	\end{array}
		\]
\caption{\label{Y}}
\end{table}

	Now we claim that $\phi$ is an instance of \pb{1-in-3-sat}, if and only if 
	$ \mA \models_X \psi $.%((\ell \subseteq v \wedge b \subseteq c \wedge a \subseteq \one) \vee(\ell \subseteq v \wedge ((b \subseteq c \wedge a \subseteq \two) \vee ( b \subseteq c \wedge a \subseteq \three)). $$
	
%	\begin{itemize}
%		\item[$\Leftarrow$]
			Let us suppose that $ \mA \models_X \psi $, i.e., there exists a partition of $X$ into three subsets $X_1$, $X_2$ and $X_3$ such that 
			\begin{align*}
				\mA & \models_{X_1}  b \subseteq c \wedge a \subseteq \one,\\
				\mA & \models_{X_2} b \subseteq c \wedge a \subseteq \two,\\
				\mA & \models_{X_3} b \subseteq c \wedge a \subseteq \three,\\
				\mA & \models_{X_1} \ell \subseteq v.
			\end{align*}
			
			We will define  an assignment $I$ over $V$ witnessing that $\phi$ is an instance of \pb{1-in-3-sat}: for $1 \leq i \leq n$ there exists a unique $s \in X$ such that $ s(v)  = v_i$. If $s \in X_1$, we set $I(v_i) = 1$ and $I(v_i) = 0$ otherwise. As $s$ is unique and the sets  $X_i$ are disjoint (because we are in the strict semantics), $I$ is well defined.
			
			We have to check that for any $1 \leq i \leq m$,  exactly one of the  variables of clause $C_i$ is evaluated to $1$ by the assignment $I$.
			
		Now 	$\mA \models_{X_1} a \subseteq \one$, $\mA \models_{X_2} a \subseteq \two$ and $\mA \models_{X_3} a \subseteq \three$ imply that every assignment $s \in X$ such that $s(a) = 1$ must be in $X_1$. Therefore $X_1(b) = \{0,1,\hdots, m\}$.	Similarly $X_2(b) = X_3(b) = \{0,1,\hdots, m\}$.

			The variable  $c$  stores the index of a clause. Because $\mA \models_{X_1} b \subseteq c$, every clause $C_i$ has an assignment $s \in X_1$ such that $ s(c)  = i$, and the same holds for the sets   $X_2$ and in $X_3$. %Similarly there is in $X_2$ and in $X_3$ an assignation which concern the clause $C_i$. $C_i$ has only three variables, one in $X_1$, one in $X_2$, one in $X_3$.
			 Thus the claim follows.
			
%		\item[$\Rightarrow$] 
		Let $I: \{v_1, \hdots, v_n\} \rightarrow \{0,1\}$ be an assignment witnessing that $\phi$ is an instance of \pb{1-in-3-sat}. Define a partition of $X$ into $X_1, X_2, X_3$ as follows. Let $s\in X$.
		
		\begin{enumerate}
		    \item\label{c1}  If $ s(v) = v_i$ and $I(v_i) = 1$, we assign $s \in X_1$. If $ s(v) = v_i$ and $I(v_i) = 0$ we assign $s \in  X_2$,
		    \item\label{c22} 	For every $1 \leq i \leq m$, $C_i = l_{i_1} \vee l_{i_2} \vee l_{i_3}$. %Without loss of generality, we may assume that $I(l_{i_1}) = 1$ and $I(l_{i_2}) = I(l_{i_3}) = 0$. 
		    Let $s \in X$ be such that $ s(c) = i$. We send the unique $s$ such that $s(\ell) = l_{i_j}$ and $I(l_j)=1$ to $X_1$ and assign exactly one of the remaining two such assignments $s$ to $X_2$ and $X_3$,%for $1\le j\le 3$, we send $s$ to $X_j$,
		    \item\label{c3} If $s(a) = k$,  we send $s$ to $X_k$. 
		\end{enumerate}
		 By \eqref{c1} and \eqref{c22} it now  clearly holds that  $\mA \models_{X_1} \ell \subseteq v$. Furthermore, by \eqref{c22} and \eqref{c3} it holds that $\mA \models_{X_k} b \subseteq c$ and    $\mA \models_{X_k} a \subseteq w_k$, for  $k = 1,2,3$, respectively.
		 %$\mA \models_{X_2} a \subseteq \two$ and $\mA \models_{X_3} a \subseteq \three$.

	%	let $s\in X$. If $ s(v) = v_i$ and $I(v_i) = 1$, we set $s \in X_1$. If $ s(v) = v_i$ and $I(v_i) = 0$ we set $s \in  X_2$. %Now clearly $\mA \models_{X_1} \ell \subseteq v$.
	%	For every $1 \leq i \leq m$, $C_i = l_{i_1} \vee l_{i_2} \vee l_{i_3}$. We can suppose that $I(l_{i_1}) = 1$ and $I(l_{i_2}) = I(l_{i_3}) = 0$. Let $s \in X$ such that $ s(c) = i$. If $s(\ell) = l_{i_j}$, for $1\le j\le 3$, we send $s$ to $X_j$.  %If  $s(\ell) = v_{i_2}$, we send $s$ to $X_2$ and if $\ell \langle s \rangle = v_{i_3}$, we send $s$ to $X_3$.  
%Hence it holds that % exactly one variable of each clause is in $X_k$, 
 %$\mA \models_{X_k} b \subseteq c$, for  $k = 1,2,3$.
%	Finally, we send every $s\in X$ such that $s(a) = k$ to $X_k$ for $k= 1,2,3$. Therefore it also holds  that $\mA \models_{X_1} a \subseteq \one$, $\mA \models_{X_2} a \subseteq \two$ and $\mA \models_{X_3} a \subseteq \three$. By the definition of $X_1$ it also clearly holds that  $\mA \models_{X_1} \ell \subseteq v$. %For  $k = 1,2,3$, exactly one variable of each clause is in $X_k$, thus $\mA \models_{X_k} b \subseteq c$.
%	Finally note that every variable evaluated to $0$ is in $X_2 \cup X_3$, hence $\mA \models_{X_2 \cup X_3} \ell \subseteq v$. Analogously, it holds that $\mA \models_{X_1} \ell \subseteq v$.
%	\end{itemize}
\end{proof}

\section{Conclusion}
On this paper we have studied the  tractability/intractability frontier  of data complexity of both quantifier-free and quantified  dependence, independence, and inclusion logic formulas. Furthermore, we defined a novel translation of  inclusion logic formulas into dual-Horn propositional formulas, and used it to show that the data-complexity of inclusion logic is in PTIME. In a paper under preparation we shall consider similar questions for quantifier-free formulas containing, in addition to dependence and independence atoms, also so-called anonymity atoms.
Although our results shed light on the tractability/intractability frontiers studied in this article, many open  questions related to the data-complexity of quantifier-free independence logic formulas remain. The general goal is to find fragments where we can prove PTIME/NP-complete dichotomy results. Our results on disjunctions of independence atoms show that the data-complexity question is more complex than merely the length of the disjunction, as is the case with dependence atoms. In a different direction,  it is an open question whether Theorem 13 holds under the strict semantics. 
%It is an interesting open question whether the translation of Proposition 16 can be generalized to hold for some interesting extensions of $\FO\subseteq$ by further dependency atoms.

\section*{Acknowledgements}
%The authors would like to thank Arne Meier for a number of corrections and useful suggestions. 
The second author was supported by the Academy of Finland grant 308712. The fourth author was supported by the Faculty of Science of the University of Helsinki and the Academy of Finland grant 322795. This project has received funding from the European Research Council (ERC) under the European Union’s Horizon 2020 research and innovation programme (grant agreement No 101020762).
\bibliographystyle{abbrv}
\bibliography{biblio}

%May 25, 2021

\end{document}